\documentclass{amsart}

\usepackage{amssymb}
\usepackage{amsfonts}
\usepackage{amscd}
\usepackage{amsmath}
\usepackage{mathtools}
\usepackage[shortlabels]{enumitem}
\usepackage{epsfig}
\usepackage{graphicx}
\usepackage{overpic}
\usepackage{verbatim}
\usepackage{latexsym}
\usepackage{eucal}
\usepackage{bbm}
\usepackage{mathrsfs}
\usepackage{mathabx}
\usepackage{dsfont}

\usepackage{xcolor}

\usepackage{mathtools, array, hyperref}

\newtheorem{theorem}{Theorem}[section]
\theoremstyle{plain}
\newtheorem{corollary}[theorem]{Corollary}
\newtheorem{lemma}[theorem]{Lemma}
\newtheorem{proposition}[theorem]{Proposition}
\theoremstyle{remark}

\newtheorem{remark}[theorem]{Remark}

\numberwithin{equation}{section}

\newcommand{\re}{\operatorname{Re}}

\newcommand{\supp}{\operatorname{supp}}

\newcommand{\Schr}{Schr\"odinger }

\def\beq{\begin{equation}}
\def\eeq{\end{equation}}
\newcommand{\bea}{\begin{eqnarray}}
\newcommand{\eea}{\end{eqnarray}}
\newcommand{\beas}{\begin{eqnarray*}}
\newcommand{\eeas}{\end{eqnarray*}}

\newcommand{\bbR}{\mathbb{R}}
\newcommand{\R}{\mathbb{R}}

\newcommand{\bbC}{\mathbb{C}}
\newcommand{\bbP}{\mathbb{P}}
\newcommand{\bbE}{\mathbb{E}}

\newcommand{\bbZ}{\mathbb{Z}}
\newcommand{\Z}{\mathbb{Z}}

\newcommand{\bbN}{\mathbb{N}}
\newcommand{\N}{\mathbb{N}}

\newcommand{\calS}{\mathcal{S}}

\newcommand{\cinf}{C^\infty}
\newcommand{\del}{\partial}
\newcommand{\vep}{\varepsilon}

\DeclarePairedDelimiter\paren{(}{)}
\DeclarePairedDelimiter\set{\{}{\}}
\DeclarePairedDelimiter\sqbrak{[}{]}
\DeclarePairedDelimiter\abs{|}{|}
\DeclarePairedDelimiter\norm{\Vert}{\Vert}
\DeclarePairedDelimiter\brak{\langle}{\rangle}

\newcommand{\Op}{\operatorname{Op}}

\begin{document}

\title[Spectral estimates for random Landau Schr\"odinger operators]{A semiclassical approach to spectral estimates for random Landau Schr\"odinger operators}


\author[D.\ Borthwick]{D.\ Borthwick}
\address{Mathematics Department, Emory University, Atlanta GA 30342, USA}
\email{dborthw@emory.edu}

\author[S.\ Eswarathasan]{S.\ Eswarathasan}
\address{Mathematics Department, Dalhousie University, Halifax, Nova Scotia, B3H 4R2, Canada}
\email{sr766936@dal.ca}

\author[P.\ D.\ Hislop]{P. D.\ Hislop}
\address{Department of Mathematics,
    University of Kentucky,
    Lexington, Kentucky  40506-0027, USA}
\email{peter.hislop@uky.edu}


\begin{abstract}
We prove spectral properties for {random} Landau Schr\"odinger operators on $L^2(\R^2)$ with bounded, random potentials supported in a square $\Lambda_L \subset \R^2$ of side length $L>0$, using semiclassical pseudodifferential calculus. The semiclassical parameter $h$ is 
the inverse of the magnetic field strength $B > 0$. By means of the Grushin method, we are led to the analysis of an effective 
Hamiltonian on $L^2 (\R)$, the principal term of which is a sum of certain compact, self-adjoint pseudodifferential operators. 
By analyzing these operators, we prove semiclassical Wegner and Minami estimates for the random Landau \Schr operator in energy intervals in the spectral bands around each Landau level.
\end{abstract}

\maketitle

\setcounter{tocdepth}{1}
\tableofcontents

\section{Introduction}\label{intro.sec}

The Landau Hamiltonian is a Schr\"odinger operator on $L^2 (\R^2)$ describing an electron restricted to the plane and moving in a constant, transverse, magnetic field with strength $B > 0$. It is well-known that the spectrum consists of a discrete set of eigenvalues $B_n$, $n \in \N_0 := \N \cup \{ 0 \}$, each of which is infinitely degenerate \cite[section 112]{LandauLiftschitz}. 
Landau Hamiltonians with random potentials play a central role in explanations of the integer quantum Hall effect \cite{bvesb94}. 
Localization properties for these random \Schr operators were proved in \cite{ch96, wm_wang97, gkGAFA}. 

In this article, we pursue the study of spectral properties of random Landau Hamiltonians using 
semiclassical pseudodifferential calculus.  We follow the broad approach of
Wang \cite{wm_wang95, wm_wang97}, based on the Grushin method described in Bellissard \cite{bell88} and 
Helffer-Sj\"ostrand \cite{HS89}; see also \cite{CR, SjZw07} for textbook accounts of the Grushin method.
Our analysis of the effective Hamiltonian obtained by the Grushin method 
provides a new and a more transparent proof of the Wegner estimate, and the first proof of a Minami estimate for random Landau Hamiltonians in the large $B$ regime. 


\subsection{Random Landau \Schr operators}\label{subsec:randomLandau1}

In $\bbR^2$, given a constant magnetic field parameter $B>0$, we define a vector potential $A(x,y) \in \R^2$ by
\[
A(x,y) = \frac{B}2 (y,-x), 
\]
The Landau Hamiltonian is the differential operator 
\[
H_0 := (-i\nabla - A)^2,
\]
which is essentially self-adjoint on $C_0^\infty (\R^2)$. It has pure point spectrum with the eigenvalues
\[
B_n := (2n +1)B,  \quad\text{for }n \in \bbN_0,
\] 
called the Landau levels, each having infinite multiplicity. {These spectral, and related, results are independent of the gauge chosen.}

We will consider the perturbation of $H_0$ by a real, bounded, scalar potential $V \in \cinf(\bbR^2)$. The resulting operator,
\begin{equation}\label{HV.def}
H_V := (-i\nabla - A)^2 + V, 
\end{equation}
is self-adjoint on the domain of $H_0$.
We assume throughout the paper that $\abs{V} \le 1$ and that $B > 1$ so that the perturbed spectrum is contained in disjoint bands associated to each Landau eigenvalue,
\[
\sigma(H_V) \subset \bigcup_{n \in \bbN_0} [B_n - 1, B_n + 1].
\]
{We refer to the interval $[B_n - 1, B_n + 1]$ as the $n^{th}$ Landau band}. 
Without loss of generality, we will focus on the eigenvalues in a single Landau band associated to a particular Landau level $n$. The spectral estimates described below apply to subintervals of $[B_n - 1, B_n + 1]$ 
that are bounded away from $B_n$ itself. 
Since the same arguments apply to intervals in $[ B_n - 1, B_n) \cup ( B_n, B_n+1]$ above or below $B_n$, we focus on eigenvalues in 
$(B_n, B_n+1]$ for simplicity.

Our results concern the case where $V$ is random potential of Anderson-type with single-site potential $v_0 \in \cinf_0(\bbR^2)$ where $\abs{v_0} \le 1$ and $\supp v_0 \in (-\tfrac12, \tfrac12)^2$.
 For $j = (j_1, j_2) \in \bbZ^2$, we define the shifted functions
\beq\label{eq:shift_ss1}
v_j(x,y) := v_0(x-j_1,y-j_2),
\eeq
with mutually disjoint supports. The full random potential associated with a finite region $\Lambda \subset \R^2$, for example say a square, is 
\beq\label{eq:rand_pot1}
V_\omega := \sum_{j \in \Lambda \cap \Z^2} \omega_j v_j
\eeq
where the vector $\omega$ given by
\[
\omega = (\omega_j:\> j \in \Lambda \cap \Z^2)
\]
denotes a sequence of independent, identically distributed (iid), random variables 
with the common distribution of each $\omega_j \in [-1,1]$ being
\[
d\bbP(\omega_j) = g(\omega_j) \>d \omega_j,
\]
for $g \in L^{\infty}[-1,1]$ and $g\ge 0$ with $\int g = 1$.
{The random potential $V_\omega$} satisfies $\abs{V_\omega} \le 1$ for all 
$\omega$. {The random potential $V_\omega$ preserves the essential  spectrum of $H_0$ and creates a discrete spectrum in the bands away from the Landau levels $[ B_n - 1, B_n) \cup ( B_n, B_n+1]$.}

The corresponding random Schr\"odinger operator $H_{V_{\omega}}$ depends explicitly on 
$\Lambda$ through the random potential. Although $\Lambda$ will be fixed for this discussion, 
the dependence of estimates on $\abs{\Lambda}$ is of primary interest for possible applications 
to the density of states and other aspects of localization.

\subsection{Main results}\label{subsec:main1} 

The Grushin method \cite{bell88, HS89} allows the eigenvalue problem 
for $H_{V_{\omega}}$ in a single fixed band to be reduced to a non-linear spectral problem involving compact operators. 
The precise statement is given in Theorem~\ref{bellissard.thm}. Note that this is not the single-band approximation used, for example, in \cite{pule}. The compact operators in 
question are defined using a (non-standard) semiclassical quantization of the potential function
$V_{\omega}$, with the semiclassical parameter defined as
\[
h := B^{-1}.
\]
The estimates are thus naturally adapted to the case where the magnetic field $B$ is large relative to the perturbing potential.

There are two eigenvalue estimates key to proofs of localization, properties of the density of states, 
and local eigenvalue statistics for random \Schr operators: the Wegner and Minami estimates. 
In the localization region of the deterministic spectrum, the Wegner estimate is an upper bound on the probability that $H_{V_\omega}$ has an eigenvalue in a given energy interval, while
the Minami estimate is an upper bound on the probability that at least two eigenvalues are   in a given energy interval. \textit{The form of these estimates 
reflects the heuristic that the contributions to the spectrum of $H_{V_\omega}$ from each lattice site behave as 
independent random variables.} 

A Wegner estimate for random Landau \Schr operators with a \emph{sign-definite} single-site potential
$v_0$ was proved in \cite[Theorem 3.1]{ch96}. In our notation, it takes the form:
\beq\label{eq:wegner_ch}
\bbP \paren[\Big]{\#\paren[\big]{\sigma(H_{V_\omega}) \cap (B_n + I)} \ge 1}
\le Ch^{-1} \abs{\Lambda} \abs{I}
\eeq
for $I \subset [b_0,1]$ with $b_0>0$, and where $C$ depends on $\norm{g}_\infty$ and $b_0$. 

One of the main contributions of \cite{wm_wang97} 
is a Wegner estimate for a \emph{non-sign-definite} single-site potential $v_0$. 
Wang's result \cite[Proposition 3.1]{wm_wang97} is
\beq\label{eq:wegner_wmw}
\bbP \paren[\Big]{\#\paren[\big]{\sigma(H_{V_\omega}) \cap (B_n + I)} \ge 1}
\le Ch^{-1} \abs{\Lambda}^2 \abs{I} .
\eeq
The key difference between \eqref{eq:wegner_ch} and \eqref{eq:wegner_wmw} is the dependence on the 
surface area $\abs{\Lambda}$. Whereas the $\abs{\Lambda}^2$ dependence of \eqref{eq:wegner_wmw} is sufficient for a proof of localization, estimate \eqref{eq:wegner_wmw} cannot be used to prove the absolute continuity of the density of states measure obtained by taking infinite surface area limit. Estimate \eqref{eq:wegner_ch} is sufficient for this and establishes the Lipschitz continuity of the integrated density of states. In Proposition \ref{prop:wegner_edge}, we prove a Wegner estimate of the form \eqref{eq:wegner_ch}  but only for energies in a small interval near the band edges. It remains an open problem to prove a bound linear in $\abs{\Lambda}$ for the case of non-sign-definite single-site potentials and general energy intervals in the localization regions. 

In Section \ref{wegner.sec}, we establish the following semiclassical version of the Wegner estimate for non-sign-definite, single-site potentials $v_0$ : 

\begin{theorem}[Wegner estimate] \label{thm:wegner_h}
There exists $h_0$ depending on $v_0$ such that for $I \subset [b_0,1]$, with $b_0>0$ and $h \le h_0$,  
\beq\label{wegner_h}
\bbP \paren[\Big]{\#\paren[\big]{\sigma(H_{V_\omega}) \cap (B_n + I)} \ge 1}
\le Ch^{-1} \abs{\Lambda} \paren*{\abs{I} + O(h^3)},
\eeq
where the constant $C$ and the $O(h^3)$ error estimate depend on $n$, $v_0$, $b_0$, $h_0$, and the probability density $g$.
\end{theorem}

We present a new and shorter proof of Wang's estimate \eqref{eq:wegner_wmw}, in Section \ref{scaling.sec}, 
which takes advantage of the fact that the effective Hamiltonian produced by the Grushin method is linear in the potential to leading order. In contrast, the proof of Theorem~\ref{thm:wegner_h}, presented in Section \ref{landau.sec}, follows more closely heuristic mentioned above. That is, the spectrum of the effective Hamiltonian can ``almost'' be treated as a union of independent random variables corresponding
to the compact operators constructed by the quantization of the individual single-site potentials. This \emph{lattice-site} approach, which is set up in section \ref{lattice.sec}, 
has the advantage of providing a very clean proof of the Wegner estimate with the correct volume scaling. 
The drawback, however, is the appearance of $h$-dependent error terms, 
resulting in an extra $O(h^3)$ on the right-hand side of \eqref{wegner_h}.

A proof of the Minami estimate, an estimate for the probability of the event that there are at least two eigenvalues in an energy interval, is more difficult. 
The Minami estimate was first proved in \cite{minami} for discrete random \Schr operators with Anderson-type 
random potentials acting on $\ell^2(\mathbb{Z}^d)$; see also \cite{CGK1}. 
The Minami estimate plays a key role in the proof that the local eigenvalue statistics in the localization region of the deterministic spectrum is a Poisson point; see also \cite{gk14}. A proof of the Minami estimate for random \Schr operators on $L^2 ( \R^d)$ has been difficult to obtain; see \cite{CGK14}. Recently, Dietlein and Elgart  \cite{de21} proved a weaker version of the Minami estimate for random Schr\"odinger operators with sign-definite potentials on $L^2 (\R^d)$ for energies near the bottom of the deterministic spectrum. This estimate is sufficient to prove Poisson eigenvalue statistics for low energy intervals. It is not clear how to extend the technique of Dietlein-Elgart to random Landau Hamiltonians.  

By the same philosophy used for the Wegner estimate in Theorem~\ref{thm:wegner_h}, we obtain a  semiclassical 
Minami estimate, i.e., valid for large $B$ relative to the random perturbation. 
However, the Minami argument is complicated by the need to control multiplicities.
This requires an extra assumption on the single-site potential $v_0$. 
This \emph{spectral gap assumption}, stated precisely in section \ref{minami.sec}, says that a certain compact operator on $L^2(\bbR)$ 
with principal symbol $v_0$ has simple spectrum with spacing that remains approximately 
uniform away from zero as $h \to 0$. Examples of potentials satisfying the spectral gap 
assumption are given in Appendix \ref{radial.sym.sec}. 
In Section \ref{minami.sec}, we establish the following:

\begin{theorem}[Minami estimate] \label{thm:minami_h}
Assume that the potential $v_0$ satisfies the spectral gap assumption stated in \S\ref{minami.sec}, with constants $\kappa, b_0$.
For $h \le h_0$ and $I \subset [b_0,1]$ with $\kappa h \le \abs{I}$,  
\[
\bbP \paren[\Big]{\#\paren[\big]{\sigma(H_{V_\omega}) \cap (B_n + I)} \ge 2}
\le Ch^{-2} \abs{\Lambda}^2 \paren*{\abs{I} + O(h^3)}^2,
\]
where the constant $C$ and the $O(h^3)$ error estimate depend on $n$, $v_0$, $b_0$, $h_0$, and the probability density $g$.
\end{theorem}

We are currently studying the applications of these methods and results to the local eigenvalue statistics associated with energies in the localization region of the band edges. 

\subsection{Acknowledgements}\label{subsec:ack1} 

The authors wish to thank BIRS for a Research in Teams 25rit034 visit in May 2025, \emph{Random Perturbations of Landau Hamiltonians and Semiclassical Analysis}.
PDH is partially supported by Simons Foundation Collaboration Grant for Mathematicians No.\ 843327. DB and PDH thank the 
Department of Mathematics \& Statistics at Dalhousie University for its warm welcome during collaborative research visits. SE also thanks
Emory University for support during a collaborative visit. SE is partially supported by NSERC Discovery Grant RGPIN/04775-2020.

\section{The effective Hamiltonian}\label{eff.H.sec}

In this section, we set up the main tool used in the analysis, namely the reduction of the spectral problem in a
Landau band by means of the Grushin method. We study the operator $H_V$ as
defined in \eqref{HV.def} since the special features of $V_\omega$ play no role in the discussion in this section. 
The results of this section require only that $V$ is in the symbol class
\beq\label{symS.def}
S(\bbR^2) := \set*{F \in \cinf(\bbR^2):\> \sup\, \abs{\partial^\alpha F} < \infty\text{ for all }\alpha = (\alpha_1,\alpha_2) \in \bbN_0^2},
\eeq
and that $\abs{V} \le 1$. 
As in \S\ref{intro.sec} we set $h = B^{-1}$ and fix a particular Landau level $n$.

For applications to the random potential case, we must keep careful track of how estimates depend
on $V$. In particular, the estimates depend on the norms 
\[
\norm{V}_{C^r_{\rm b}} := \sum_{\abs{\alpha} \le r} \sup\,\abs{\del^\alpha V},
\]
but not on the support of $V$.

From \cite{bell88, HS89} we paraphrase the following result, whose proof is sketched in Appendix~\ref{grushin.sec}:
\begin{theorem}[Bellissard, Helffer-Sj\"ostrand]\label{bellissard.thm}
There is a constant $h_0$, depending only on the Landau level $n$ and 
$\norm{V}_{C^r_{\rm b}}$ for some $r$, such that $h \le h_0$ there exists a family of zeroth order 
pseudodifferential operators $Q_V(\mu)$ acting on $L^2(\bbR)$ such that 
\[
B_n + \mu \in \sigma(H_V) \quad\Longleftrightarrow\quad 0 \in \sigma(Q_V(\mu))
\]
for $\mu \in [-1,1]$. The family $Q_V(\mu)$ depends analytically on $\mu$.
\end{theorem}

The family of operators $Q_V(\mu)$ arises as the effective Hamiltonian in an application of the Grushin method. 
For our arguments, we will take the power series expansion of $Q_V(\mu)$ from \cite[\S2.3]{HS89} as a starting point:
\begin{equation}\label{QV.exp}
Q_V(\mu) := -\mu + \sum_{j=0}^\infty (-1)^j h^j R_n^* W \paren[\big]{E_0(\mu) W}^j R_n.
\end{equation}
where the operators $W, R_n,$ and $E_0$ are defined as follows.  Here $W$ is an operator on $L^2(\bbR^2)$ defined by the following quantization of $V$,
\begin{equation}\label{W.def}
\begin{split}
Wu(x,y) &:= \frac{1}{4\pi^2} \int_{\bbR^4} e^{i(x-x')\xi + i(y-y')\eta} V\paren*{\frac{y+y'}2 + h^{\frac12} \frac{x+x'}2, h\eta - h^{\frac12}\xi} \\
&\hskip1in\times u(x',y')\> dx'\>dy'\>d\xi\>d\eta.
\end{split}
\end{equation}
In terms of the $L^2$-normalized harmonic oscillator eigenfunctions 
\begin{equation}\label{psil.def}
\psi_l(x) := \frac{\pi^{-\frac14}}{\sqrt{2^ll!}} (\del_x-x)^l e^{-x^2/2},
\end{equation}
$R_l$ denotes the operator $L^2(\bbR) \to L^2(\bbR^2)$ given by 
\begin{equation}\label{Rl.def}
R_l : v(y) \mapsto \psi_l(x)v(y),
\end{equation}
with the adjoint, 
\[
R_l^* : u(x,y)\mapsto \int_{\bbR} \psi_l(x) u(x,y)\>dx.
\]
Note that $R_l^* R_k  = \delta_{lk} \ I$, on $L^2(\R)$. 
Finally, with $n$ the fixed Landau level, the operator $E_0$ is defined on $L^2(\bbR^2)$ as
\begin{equation}\label{E0.def}
E_0(\mu) := \sum_{l\ne n} \frac{1}{2(l-n) - h\mu} R_l R_l^*.
\end{equation}

Note that $\norm{R_n} = 1$ and $\norm{E_0(\mu)} = (2-h\mu)^{-1}$ for $\abs{h\mu} < 2$. Since 
$V\in S(\bbR^2)$, the operator $W$ satisfies
\beq\label{W.op.bnd}
\norm{W} \le C \norm{V}_{C^r_{\rm b}}
\eeq
for some $r$ by \cite[Theorem~4.23]{Zw12}, where $C$ is independent of $h < 1$.
Therefore for $h_0$ sufficiently small the series \eqref{QV.exp} is convergent in the operator topology.
We review the construction that leads to \eqref{QV.exp} in Appendix~\ref{grushin.sec}.

For future reference, define the harmonic oscillator operator $L := -\del_x^2 + x^2$, so that
the states $\psi_n$ satisfy 
\begin{equation}\label{L.eigv}
L\psi_n = (2n+1)\psi_n. 
\end{equation}
We will also make use of the raising and lowering operators,
\[
J_\pm := \pm \del_x - x,
\]
which satisfy
\[
L = J_+ J_-+1 = J_- J_+ -1.
\]
The harmonic oscillator states are then related by
\[
J_+ \psi_k = \sqrt{2(k+1)} \psi_{k+1}, \quad J_- \psi_k = \sqrt{2k}\psi_{k-1}.
\]

Our applications of Theorem~\ref{bellissard.thm} are based on approximations of 
$Q_V(\mu)$ as $h \to 0$. To extract these approximations from
the power series in \eqref{QV.exp}, we need to account for the $h$-dependence of the operators
$W$ and $E_0$. The terms in the expansion are expressed in terms of the semiclassical 
Weyl quantization of a symbol $a \in S(\bbR^2)$,
\begin{equation}\label{widehat.def}
\widehat{a}u(y) := \frac{1}{2\pi h} \int_{\bbR^2} e^{i(y-y')\eta/h} a\paren*{\frac{y+y'}2, \eta} u(y') dy'd\eta.
\end{equation}
We refer to \cite[\S 4.1]{Zw12} for a textbook presentation of semiclassical quantization. 

\begin{proposition}\label{Qexpand.prop}
For $h \le h_0$, the effective Hamiltonian $Q_V(\mu)$ in \eqref{QV.exp} admits the expansion
\[
Q_V(\mu) = \widehat{V} - \mu + h \widehat{V_1} + h^2 \widehat{V_2} + h^3Q_V^{(3)}(\mu),
\]
where 
\[
\begin{split}
V_1 &= \frac{2n+1}{4} \Delta V, \\
V_2 &= \frac14 (\nabla V)^2 + \sum_{\substack{\abs{\alpha} = 4 \\\alpha_j \text{even}}} c_\alpha \del^\alpha V,
\end{split}
\]
where $\alpha = (\alpha_1,\alpha_2) \in \N_0^2$, and the remainder satisfies an estimate 
\[
\norm{Q_V^{(3)}(\mu)} \le C
\]
with $C$ depending only on $n$ and $\norm{V}_{C^r_{\rm b}}$ for some $r$. The coefficients $c_\alpha$ are defined in \eqref{c.alpha} and thereafter. 
\end{proposition}

\begin{proof}
To analyze the $j=0$ term of the series in \eqref{QV.exp}, we start with the Taylor expansion of $V$, 
using a multi-index $\alpha = (\alpha_1,\alpha_2)$,
\begin{equation}\label{V.taylor}
\begin{split}
V(y+h^{\frac12}x, h\eta - h^{\frac12} \xi) &= \sum_{\abs{\alpha} \le 5} \frac{1}{\alpha!} h^{\abs{\alpha}/2} 
x^{\alpha_1}(-\xi)^{\alpha_2} \del^\alpha V(y,h\eta) \\
&\qquad+ h^3 r_6(x,y,\xi,\eta),
\end{split}
\end{equation}
where $r_6 \in \brak{(x,\xi)}^6 S(\bbR^4)$, with $ \brak{\eta} := ( |\eta|^2 + 1)^{\frac{1}{2}}$. The contribution to the $j=0$ term $R_n^* W R_n$ in the expansion of  $Q_V(\mu)$ for the index $\alpha$ has 
symbol $\del^\alpha V(y,h\eta)$ with the coefficient
\begin{equation}\label{c.alpha}
c_\alpha = \frac{1}{\alpha!} \int_{\bbR^3} e^{i(x-x')\xi} \paren*{\frac{x+x'}{2}}^{\!\alpha_1}(-\xi)^{\alpha_2} \psi_n(x)
\psi_n(x')\>dx\>dx' \>d\xi
\end{equation}
If $\alpha_2$ is odd, then the integrand in \eqref{c.alpha} is odd with respect to the change of variables $(x,x',\xi) \mapsto (x',x,-\xi)$,
and therefore $c_\alpha = 0$ in this case. Furthermore, if $\alpha_2$ is even and $\alpha_1$ is odd, then the integrand is odd 
with respect to $(x,x',\xi) \mapsto -(x,x',\xi)$.
It follows that $c_\alpha = 0$ unless both $\alpha_1$ and $\alpha_2$ are even. 

Obviously $c_0 = 1$, and the contributions from $\abs{\alpha} =2$ are easily computed.
The harmonic oscillator eigenvalue equation \eqref{L.eigv} gives
\[
c_{(2,0)} = c_{(0,2)} = \frac{2n+1}{4}.
\]
The coefficients $c_\alpha$ for $\alpha = (4,0)$, $(2,2)$, and $(0,4)$ could similarly be computed explicitly from
\eqref{c.alpha}, but the exact form is not important here, as they are constants that depend only on $n$.

To estimate the remainder term in \eqref{V.taylor}, 
note that $\brak{(x,\xi)}^{-6} r_6 \in S(\mathbb{R}^4)$ 
and that $(L+1)\otimes I$ is the quantization of $\brak{(x,\xi)}^2$.
A basic symbol composition result, e.g.~\cite[Theorem~4.18]{Zw12}, shows that $((L+1)^{-3}\otimes I) \Op(r_6)$ has a symbol
in $S(\bbR^4)$, where $\Op(\cdot)$ denotes the classical Weyl quantization of a symbol on $\bbR^4$. 
It follows that 
\[
\norm*{\paren[\big]{(L+1)^{-3}\otimes I} \Op(r_6)}_{L^2(\bbR^2) \to L^2(\bbR^2)} \le C\norm{V}_{C^r_{\rm b}},
\]
by \cite[Theorem~4.23]{Zw12}. By \eqref{L.eigv}, and reminding ourselves that $\|R_n\|=1$, this yields a bound 
\begin{equation}\label{r6.bound}
\norm*{R_n^* \Op(r_6) R_n}_{L^2(\bbR) \to L^2(\bbR)} \le Cn^3 \norm{V}_{C^r_{\rm b}}.
\end{equation}

We can now apply these observations to compute the contribution to \eqref{QV.exp} from the $j=0$ term. 
Applying \eqref{r6.bound} to \eqref{V.taylor} and using the calculations of coefficients above reduces 
the $j=0$ component to
\[
h\frac{2n+1}4 \widehat{\Delta V} + h^2 \sum_{\substack{\abs{\alpha} = 4 \\\alpha_j \text{even}}} c_\alpha \widehat{\del^\alpha V} + O(h^3).
\]

Now consider the $j=1$ term of the series in \eqref{QV.exp}, 
\begin{equation}\label{j1.term}
-hR_n^* W E_0(\mu)W R_n,
\end{equation}
which is nominally $O(h)$. As in \eqref{V.taylor}, we expand
\[
\begin{split}
V(y+h^{\frac12}x, h\eta - h^{\frac12} \xi) &= \sum_{\abs{\alpha} \le 3} \frac{1}{\alpha!} h^{\abs{\alpha}/2} 
x^{\alpha_1}(-\xi)^{\alpha_2} \del^\alpha V(y,h\eta) \\
&\qquad+ h^{2} r_4(x,y,\xi,\eta).
\end{split}
\]
The $\alpha=0$ terms do not contribute to \eqref{j1.term}, because 
\begin{equation}\label{ER.cancel}
E_0(\mu) R_n = 0,\qquad R_n^* E_0(\mu) = 0,
\end{equation}
since $E_0$ projects onto the orthogonal complement of $\psi_n$.
The leading term in \eqref{j1.term} thus comes from the linear terms in the expansions, and is given by
\[
-h^2 R^+_n \Op(x \del_1V - \xi \del_2 V) E_0(\mu) \Op(x \del_1V - \xi \del_2 V)R^-_n.
\]
As above, the computation of coefficients reduces to harmonic oscillator calculations, which
we can facilitate using
\[
\Op(x) = -\frac{1}{2} (J_+ + J_-)\otimes I, \qquad \Op(\xi) = \frac{i}{2} (J_+ - J_-)\otimes I.
\] 
For example, the coefficient of $(\del_1V)^2$ is 
\[
\begin{split}
&\frac14 R_n^* \paren[\big]{(J_+ + J_-)\otimes I} E_0(\mu) \paren[\big]{(J_+ + J_-)\otimes I} R_n \\
&\qquad= 
\frac14 \paren*{\frac{1}{2-h\mu} \brak{\psi_{k+1}, J_+\psi_k}^2 + \frac{1}{-2-h\mu} \brak{\psi_{k-1}, J_-\psi_k}^2} \\
&\qquad= 
\frac14 \paren*{\frac{2(n+1)}{2-h\mu} - \frac{2n}{2+h\mu}} \\
&\qquad=
\frac14 + O(h).
\end{split}
\]
The coefficient of the mixed term $\del_1V\cdot\del_2V$ vanishes, by the calculation,
\[
\begin{split}
& \frac{1}{4} R_n^* \paren[\big]{(J_+ + J_-)\otimes I} \paren*{\frac{\sqrt{2(n+1)}}{2 - h \mu} R_{n+1} -  \frac{\sqrt{2n}}{-2 - h \mu} R_{n-1}} \\
&\qquad + \frac{1}{4} R_n^* \paren[\big]{(J_+ - J_-)\otimes I} \paren*{\frac{\sqrt{2(n+1)}}{2 - h \mu} R_{n+1} +  \frac{\sqrt{2n}}{-2 - h \mu} R_{n-1}} = 0.
\end{split}
\]
The coefficient of $(\del_2V)^2$ is given by
\[
\begin{split}
-\frac14 R_n^* \paren[\big]{(J_+ - J_-)\otimes I} E_0(\mu) \paren[\big]{(J_+ - J_-)\otimes I} R_n &= \frac14 \paren*{\frac{2(n+1)}{2-h\mu} - \frac{2n}{2+h\mu}} \\
&= \frac{1}{4} + O(h).
\end{split}
\]

The coefficients of terms in \eqref{j1.term} of the form $\del^\alpha V D^\beta V$ vanish unless $\abs{\alpha}$ and $\abs{\beta}$ have the same parity, for the same reasons as before. Hence the remaining terms of this type have $\abs{\alpha} + \abs{\beta}$ equal to either $4$ or $6$, and thus 
contribute to $Q_V(\mu)$ terms of order $h^3$ or lower. 

Again, the argument used above for \eqref{r6.bound} yields the estimate
\[
\norm*{R_n^* \Op(r_4) R_n}_{L^2(\bbR) \to L^2(\bbR)} \le Cn^2 \norm{V}_{C^r_{\rm b}}.
\]
Applying these results to \eqref{j1.term} yields the contribution from the $j=1$ term as
\[
\frac{h^2}4 \widehat{(\nabla V)^2} + O(h^3).  
\]

Finally, consider the remaining terms in \eqref{QV.exp} with $j \ge 2$. 
For these terms we only need to apply the leading approximation,
\[
V(y+h^{\frac12}x, h\eta - h^{\frac12} \xi) = V(y,h\eta) + h^{1/2} r_1(x,y,\xi,\eta),
\]
to the the outer pair of $W$ operators. Then by \eqref{ER.cancel} we have
\[
\begin{split}
&R_n^* W \paren[\big]{E_0(\mu) W}^j R_n \\
&\qquad = 
R_n^* \paren*{\widehat{V} + h^{1/2} \Op(r_1)} \paren[\big]{E_0(\mu) W}^{j-1} E_0(\mu)  \paren*{\widehat{V} + h^{1/2} \Op(r_1)} R_n \\
&\qquad = h^{1} R_n^* \Op(r_1) \paren[\big]{E_0(\mu) W}^{j-1} E_0(\mu) \Op(r_1) R_n.
\end{split}
\]
As above, we note that $\brak{(x,\xi)}^{-1}r_1 \in S(\bbR^4)$ and so $((L+1)^{-1/2}\otimes I) \Op(r_1)$, which gives
an estimate
\[
\norm{\Op(r_1) R_n}_{L^2(\bbR^2) \to L^2(\bbR^2)} \le Cn^{1/2} \norm{V}_{C^r_{\rm b}}.
\]
Using the bound \eqref{W.op.bnd} and the fact that $\norm{E_0(\mu)} \le 1$ for $h <1$ yields
\[
\norm*{R_n^* W \paren[\big]{E_0(\mu) W}^j R_n} \le Ch \norm{V}_{C^r_{\rm b}}^{j+1}
\]
for $j\ge 2$, where $C$ depends only on $n$. For $h \le h_0$ with $h_0$ sufficiently small, depending on $\norm{V}_{C^r_{\rm b}}$,
we can sum these estimates to obtain the bound,
\[
\sum_{j=2}^\infty h^j \norm*{R_n^* W \paren[\big]{E_0(\mu) W}^j R_n} = O(h^3).
\] 
\end{proof}

Proposition \ref{Qexpand.prop} is our main technical result on the effective Hamiltonian that will be used extensively in the following sections. 


\section{Spectral analysis of single-site operators}\label{lattice.sec}

The leading terms in the effective Hamiltonian $\widehat{V}$ and $\widehat{V}_1$, given in Proposition \ref{Qexpand.prop}, are linear in $V$, and thus can be written as a sum over operators attached to individual lattice sites in $\mathbb{Z}^2$. Our goal in this section is to develop the results that will allow us to 
compare the spectral properties of the effective Hamiltonian to the spectrum of these single lattice site operators. We present the general theory in this section and apply the results to the random Landau model in section \ref{landau.sec}.

We start with a real-valued symbol $a_0 \in \cinf_0(\bbR^2)$ satisfying 
\[
\supp a_0 \subset (-\tfrac12,\tfrac12)^2.
\]
For $j \in \bbZ^2$ we define the shifted symbols
\[
a_j(x,y) := a_0(x-j_1,y-j_2)
\]
and consider the semiclassical Weyl quantization of these symbols acting on $L^2(\bbR)$,
\[
\widehat{a}_j u(y) := \frac{1}{2\pi h} \int_\bbR e^{i(y-y')\eta/h} a_j(\tfrac{y+y'}2,\eta) u(y')\>dy'\>d\eta.
\] 
\begin{lemma}\label{aj.spec.lemma}
The operators $\widehat{a}_j$ are self-adjoint and Hilbert-Schmidt, and 
the spectrum $\sigma(\widehat{a}_j)$ is independent of $j$.
\end{lemma}
\begin{proof}
The Weyl quantization of a real-valued symbol yields a self-adjoint operator. 
The Hilbert-Schmidt property follows from the observation that the integral kernel of $\widehat{a}_j$ lies in $\calS(\bbR^2)$.
The fact that the operators have spectrum independent of $j$ follows from the unitary map 
\[
u(y) \mapsto e^{-ij_2y/h} u(y+j_1).
\]
which shows that $\widehat{a}_j$ is unitarily equivalent to $\widehat{a}_0$.
\end{proof}

Our primary goal here is to relate the spectrum of the operator
\beq\label{eq:sss1}
A:= \sum_{j\in \widetilde{\Lambda}} \widehat{a}_j
\eeq
to the spectrum of the individual operators $\widehat{a}_j$, where  $\widetilde{\Lambda} := \Lambda \cap \Z^2$.  In the next section we will let each 
$a_j$ depend on a random parameter, but the randomness is not relevant for the estimates in this section.
The region $\Lambda \subset \R^2$ is fixed for this entire discussion, but it is important for later applications that certain estimates are independent of $\Lambda$.

The assumption of disjoint supports yields the following:
\begin{lemma}\label{site.sep.lemma}
For all $N>0$, there exists $C_N$ depending only on $a_0$ such that
\[
\norm*{\widehat{a}_i \widehat{a}_j} \le C_N h^N \abs{i-j}^{-N}
\]
for $h < 1$.
\end{lemma}

\begin{proof}
For all $N>0$, the mixed-term upper bound of \cite[Theorem 4.22]{Zw12} gives an estimate
\[
\norm*{\widehat{a}_i \widehat{a}_j } \le C'_N \abs{i-j}^{-N},
\]
with $C'_N$ independent of $h$.  Furthermore, the standard result on Moyal products $a \# b$ for symbols 
with disjoint supports \cite[Theorem 4.12]{Zw12} yields an estimate
\[
\norm*{\widehat{a}_i \widehat{a}_j} \le C''_N h^N.
\]  
Taking the square root of the product of these upper bounds gives our desired result.
\end{proof}

To isolate the individual single-site operators, we introduce a cutoff $\chi_0 \in \cinf_0(\bbR^2)$ with 
$\supp \chi_0 \in (-\tfrac12, \tfrac12)^2$ and $\chi_0 =1$ on some open neighborhood of $\supp a_0$. The operator 
\begin{equation}\label{Pj.def}
P_j := \widehat{\chi_j}
\end{equation}
will serve as a localization onto site $j$. The product $P_iP_j$ satisfies the estimate of Lemma~\ref{site.sep.lemma}.
We also define the remainder operator
\[
Q := 1 - \sum_{j \in \widetilde{\Lambda}} P_j.
\]
The operator $Q$ is the quantization of the symbol $1 - \sum \chi_j$, which has support disjoint from that of any of the symbols $a_j$.

To simplify the estimates below, we use $O(h^\infty)$ to indicate a bound by $C_Nh^N$ for all $h<1$ and 
$N>0$, where $C_N$ is allowed to depend only on the symbol $a_0$ and not on $\Lambda$. 

\begin{lemma}\label{site.sep.lemma2}
Let $A$ be the operator sum in \eqref{eq:sss1}. For each $j \in \Lambda$,
\[
P_j A  = \widehat{a}_j + O(h^\infty),
\]
and the remainder $Q$ satisfies
\[
\norm{QA}  = O(h^{\infty}).
\]
\end{lemma}

\begin{proof}
For $j \in \Lambda$, applying $P_j$ to $A$ gives
\[
P_j A  = \sum_{k \in \widetilde{\Lambda}} P_j \widehat{a}_k.
\]
Lemma~\ref{site.sep.lemma} gives an estimate for any $N>0$,
\[
\norm*{P_j \widehat{a}_k} \le C_N h^N \abs{j-k}^{-N}.
\]
Assuming $N>2$, we can sum the error terms over all $j \ne k$ in $\bbZ^2$ to obtain a bound independent of $\Lambda$:
\[
\sum_{j\ne k}\norm*{P_j \widehat{a}_k} = O(h^\infty).
\]
This leaves the diagonal term $P_j \widehat{a}_j$. Since $1-P_j$ is the quantization of a symbol with support disjoint from that of $a_j$,
we can replace $P_j \widehat{a}_j$ by $\widehat{a}_j$ with an error estimate of the form
\[
\norm*{\widehat{a}_j-P_j\widehat{a}_j}  = O(h^\infty).
\]
The estimate of $\| QA \|$ follows from the disjoint supports of the symbols, as in 
\cite[Theorem 4.25]{Zw12}.
\end{proof}

After these preliminaries concerning the  localization of the operators $\widehat{a}_j$, we study the spectrum of $A$. For simplicity of notation, we write $[\mu \pm \delta]$ for the closed interval $[ \mu - \delta, \mu + \delta]$, and for an interval $I = [a,b] \subset \R$, we write $[I \pm \delta]$ for the closed interval $[a - \delta, b+\delta]$ obtained by enlarging $I$ by $\delta$.   

\begin{proposition}\label{pre.wegner.prop}
For $\mu_0 \ge b_0 >0$ and $0<\delta<b_0/2$
depending on $b_0$ and $a_0$ but not on $\Lambda$, such that if 
\begin{equation}\label{aj.site.hyp}
\sigma(\widehat{a}_j)\cap [\mu_0\pm (\delta + r(h))] = \emptyset\quad\text{for all }j \in \Lambda,
\end{equation}
then 
\[
\sigma(A) \cap [\mu_0\pm\delta] = \emptyset.
\]
\end{proposition}
\begin{proof}
We derive the result by contrapositive. 
Suppose there exists $\lambda \in \sigma(A)$ with $\abs{\lambda - \mu_0} \le \delta$. 
Let $\psi$ be a corresponding eigenvector, 
\[
(A-\lambda)\psi = 0\quad\text{with }\norm{\psi}=1.
\]
Setting $\psi_j := P_j\psi$ gives a decomposition
\[
\psi = \sum_{j \in \Lambda} \psi_j + Q\psi.
\]
The remainder term can be estimated by applying $Q$ to the eigenvalue equation, yielding an estimate
\[
\norm{Q\psi} = \frac{1}{\lambda} \norm[\big]{QA\psi} \le \frac{2}{b_0} O(h^\infty),
\]
by Lemma~\ref{site.sep.lemma2}. We will write this as 
\begin{equation}\label{psi.exp}
\psi = \sum_{j \in \widetilde{\Lambda}} \psi_j + O(h^\infty),
\end{equation}
where here and for the remainder of the proof, $O(h^\infty)$ denotes an estimate by $C_Nh^N$ for any $N$,
with $C_N$ depending only on $a_0$ and $b_0$. 

Applying $A$ to $\psi$ thus gives
\begin{equation}\label{Aw.jk}
A\psi = \sum_{j,k \in \widetilde{\Lambda}} \widehat{a}_j \psi_k + O(h^\infty).
\end{equation} 
By Lemma~\ref{site.sep.lemma}, we have estimates
\[
\norm*{\widehat{a}_j P_k} \le C_Nh^n \abs{j-k}^{-N}.
\]
Hence for $N>2$ the estimates of the off-diagonal terms in \eqref{Aw.jk} can be summed over all $j \ne k \in \bbZ^2$ 
to obtain
\[
A\psi = \sum_{j \in \widetilde{\Lambda}} \widehat{a}_j \psi_j + O(h^\infty),
\]
with estimates independent of $\Lambda$.

Using \eqref{psi.exp} and \eqref{Aw.jk} along the eigenvalue equation gives an estimate
\[
0 = \norm[\big]{(A - \lambda)\psi}^2 
= \sum_{j,k\in \widetilde{\Lambda}} \brak[\big]{(\widehat{a}_j - \lambda)\psi_j, (\widehat{a}_k - \lambda)\psi_k} + O(h^\infty)
\]
As above, estimates of the off-diagonal terms in the inner products can be summed to produce error estimates independent of $\Lambda$,
yielding
\begin{equation}\label{norm.ajpsij}
\sum_{j\in \widetilde{\Lambda}} \norm[\big]{(\widehat{a}_j - \lambda)\psi_j}^2 = O(h^\infty).
\end{equation}
With similar off-diagonal estimates we can deduce from \eqref{psi.exp} that
\begin{equation}\label{norm.psij}
\sum_{j \in \widetilde{\Lambda}} \norm{\psi_j}^2 = 1 + O(h^\infty).
\end{equation}

By the spectral theorem, 
\[
\norm[\big]{(\widehat{a}_j - \lambda)\psi_j} \ge d(\lambda,\sigma(\widehat{a}_j)) \norm{\psi_j}.
\]
Hence, from \eqref{norm.ajpsij} and \eqref{norm.psij} we can deduce that 
\[
\min_{j\in\Lambda} d(\lambda,\sigma(\widehat{a}_j)) = O(h^\infty).
\]
The function $r(h)$ is now chosen to match the estimate on the right side.

To summarize, we have shown that the existence of 
$\lambda \in \sigma(A_\omega)$ with $\abs{\lambda - \mu_0} \le \delta$ implies that
\[
\min_{j\in\Lambda} d(\lambda,\sigma(\widehat{a}_j)) \le r(h).
\]
The result follows, since the hypothesis \eqref{aj.site.hyp} implies that
$d(\lambda, \sigma(\widehat{a}_j)) > r(h)$ for all $j \in \Lambda$.
\end{proof}

\begin{proposition}\label{pre.minami.prop}
For $\mu_0 \ge b_0 >0$ and $0<\delta<b_0/2$, there exists a function $r(h) = O(h^\infty)$, 
depending on $b_0$ and $a_0$ but not on $\Lambda$, such that if 
\begin{equation}\label{aj.site.hyp1}
\sigma(\widehat{a}_j)\cap [\mu_0\pm (\delta + r(h))] = \emptyset\quad\text{for all but one }j \in \Lambda,
\end{equation}
and  for the exceptional case $\sigma(\widehat{a}_j)\cap [\mu_0\pm (\delta + r(h))]$ contains a single simple eigenvalue,
then $A$ has at most one simple eigenvalue in the range $[\mu_0\pm\delta]$.
\end{proposition}
\begin{proof}
For convenience let us assume that the special case occurs at $j=0$. 
With $r(h)$ defined as in Proposition~\ref{pre.wegner.prop},
suppose that $\widehat{a}_0$ has a simple eigenvalue $\mu$ with $\abs{\mu - \mu_0} \le \delta + r(h)$. Let $\phi \in L^2(\bbR)$ 
be the corresponding eigenfunction with $\norm{\phi}=1$. Note that $\phi$ has the same localization properties as $\widehat{a}_0$, 
in terms of the operator $P_0$ defined as in \eqref{Pj.def}. That is,
\[
\norm[\big]{(1- P_0)\widehat{a}_0} = O(h^\infty)
\]
by \cite[Theorem 4.25]{Zw12}, because the symbols have disjoint supports.
From the eigenvalue equation $(\widehat{a}_0 - \mu)\phi = 0$ we can thus deduce that
\begin{equation}\label{p0.phi.loc}
\norm*{(1-P_0)\phi} = O(h^\infty).
\end{equation}

Now consider the rank-one perturbation
\[
\widehat{a}_0' := \widehat{a}_0 - \mu \phi\otimes \overline{\phi},
\]
which moves the eigenvalue $\mu$ to $0$ without altering the rest of the spectrum.
In particular, $\widehat{a}_0'$ has no eigenvalues in the range $[\mu_0\pm (\delta+r(h))]$.
By \eqref{p0.phi.loc}, the argument from Proposition~\ref{pre.wegner.prop} applies to the perturbed operator
\[
A' := A - \mu \phi\otimes \overline{\phi}.
\]
Hence $A'$ has no eigenvalues in the range $[\mu_0\pm \delta]$.

Since the positive eigenvalues of $A'$ and $A$ may be computed using the min-max principle, 
we can argue as in the Weyl inequality for matrices. Label the positive eigenvalues of $A$ in decreasing order
as $\lambda_k$ and those of $A'$ as $\lambda_k'$. The fact that $\mu>0 $ implies that
\[
\lambda_k' \le \lambda_k.
\]
On the other hand,
\[
\begin{split}
\lambda_k' &= \min_{w_1,\dots,w_{k-1}} \paren*{\max_{u \perp \set{w_1,\dots,w_{k-1}}}
\frac{\brak{u,A u} - \mu\abs{\brak{u,\phi}}^2}{\norm{u}^2}}\\
&\ge \min_{w_1,\dots,w_{k-1}} \paren[\bigg]{\max_{u \perp \set{w_1,\dots,w_{k-1},\phi}}
\frac{\brak{u,A u}}{\norm{u}^2}} \\
&\ge \min_{w_1,\dots,w_{k}} \paren[\bigg]{\max_{u \perp \set{w_1,\dots,w_k}}
\frac{\brak{u,A u}}{\norm{u}^2}} \\
&= \lambda_{k+1}.
\end{split}
\]
Hence we have the interlacing result 
\[
\lambda_k' \le \lambda_k \le \lambda_{k-1}'.
\]
Since $A'$ has no eigenvalues in the range $[\mu_0\pm\delta]$, $A$ can have at most one simple eigenvalue in this range.
\end{proof}

\section{Return to the Landau model: proofs of the main theorems}\label{landau.sec}

As noted in the introduction, the Wegner and Minami estimates reflect the heuristic that the contributions 
to the spectrum of $H_{V_\omega}$ from each lattice site behave as independent random variables. 
In this section we use the results of  \S\ref{lattice.sec} to turn this intuition 
into a proof, at the cost of error arising from the remainder term in the approximation given in Proposition~\ref{Qexpand.prop}. 

 As above, we write $\widetilde{\Lambda} := \Lambda \cap \Z^2$.  The random potential $V_\omega$ is described in \eqref{eq:rand_pot1}, with
\[
V_\omega := \sum_{j \in \widetilde{\Lambda}} \omega_j v_j.
\]
The support requirements for $v_j$ are loosened in \S\ref{scaling.sec}, but for this section we reiterate the support 
assumption made in the introduction, 
\[
\supp v_0 \subset (-\tfrac12,\tfrac12)^2.
\]
This ensures that the shifted potentials $v_j$ have disjoint supports. 
As before, the distribution of each $\omega_j \in [-1,1]$ is given by
\[
d\bbP(\omega_j) = g(\omega_j) d\omega_j,
\]
but for this section we can assume that $g \in L^\infty[-1,1]$, with no additional regularity required.

\subsection{Semiclassical Wegner estimate}\label{wegner.sec}
 
The Wegner estimate of Theorem~\ref{thm:wegner} due to Wang \cite{wm_wang97} has the correct dependence on interval width, but with volume scaling
$\abs{\Lambda}^2$ rather than the expected $\abs{\Lambda}$. The estimates from \S\ref{lattice.sec} allow us to prove Theorem~\ref{thm:wegner_h}, 
with the correct volume scaling at the cost of introducing a $B = h^{-1}$-dependent error in the width of the interval (see, however Proposition \ref{prop:wegner_edge}). We mention that these results do not require that $v_0$ has a fixed sign, unlike \cite{ch96}. 

In the notation of Proposition~\ref{Qexpand.prop}, with $V_\omega$ replacing $V$, we define the operator
\[
A_\omega = \widehat{V_\omega} + h \widehat{V_{\omega,1}} + h^2 \widehat{V_{\omega,2}},
\]
so that 
\[
Q_{V_\omega}(\mu) = A_\omega - \mu + O(h^3).
\]
Because of the disjoint supports, we can write $A_\omega$ as a sum over individual site operators,
\[
A_\omega = \sum_{j \in \widetilde{\Lambda}} \widehat{a}_j(\omega_j),
\]
where the local symbol functions are given by
\begin{equation}\label{aw.def}
a_j(\omega_j) = \omega_j\paren[\Bigg]{v_j + \frac{2n+1}4h\Delta v_j + h^2 \sum_{\abs{\alpha} = 4}c_\alpha \del^\alpha v_j} 
+ \frac{h^2}4 \omega_j^2 (\nabla v_j)^2.
\end{equation}
Note that the semiclassical quantization \eqref{widehat.def} yields operators that depend on $h$ even if the symbol 
does not. In our case, we have in addition some explicit $h$-dependence in the symbols. By the unitary equivalence from
Lemma~\ref{aj.spec.lemma}, probabilities involving the spectrum $\sigma(\widehat{a}_j(\omega_j))$ 
are independent of $j$. 
 
\begin{lemma}\label{site.wegner.lemma}
For $\mu_0 \ge b_0 >0$ and $0<\epsilon<b_0/2$, there is a constant $C$ depending only on 
$a_0$, $b_0$ and the probability density function $g$, 
such that
\[
\bbP\paren[\Big]{\#\paren[\big]{\sigma(\widehat{a}_j(\omega_j)) \cap [\mu_0\pm \epsilon]} \ge 1} \le Ch^{-1}\epsilon.
\]
\end{lemma}
\begin{proof}
Since the spectrum is independent of $j$, it suffices to consider the case $j=0$. We may also restrict our
attention to the case $\omega_j\ge 0$, since a similar argument can be applied to $\omega_j \le 0$.
For $\omega \in [0,1]$, let $\mu_k(\omega)$ denote the $k$th positive eigenvalue of $\widehat{a}_0(\omega)$, in decreasing order. 
The min-max principle implies that $\mu_k(\cdot)$ is a continuous function of $\omega$ for each $k$. 


Let us make the dependence of \eqref{aw.def} on $\omega$ explicit by writing
\[
a_0(\omega) = \omega p_0 + \omega^2q_0
\]
where
\begin{equation}\label{p0.def}
p_0 = v_0 + \frac{2n+1}4h\Delta v_0 + h^2 \sum_{\abs{\alpha} = 4}c_\alpha \del^\alpha v_0
\end{equation}
and
\[ 
q_0 = \frac{h^2}4 (\nabla v_0)^2.
\]
Let $\nu_k(\omega)$ denote the $k$th positive eigenvalue of $\widehat{p}_0 + \omega\widehat{q}_0$, so that
\[
\mu_k(\omega) = \omega\nu_k(\omega).
\] 
By the min-max principle, $\nu_k(\omega)$ satisfies the estimate
\begin{equation}\label{nuk.lip}
\abs*{\nu_k(\omega) - \nu_k(\omega')} \le \norm{\widehat{q}_0} \,\abs{\omega-\omega'}.
\end{equation}
Since $\norm{\widehat{q}_0} = O(h^2)$, we can assume that $h_0$ is chosen so that 
$\norm{\widehat{q}_0} \le b_0/4$ for $h \le h_0$
Then, by \eqref{nuk.lip} there exists some finite $N$ such that the eigenvalues satisfy
\[
\nu_k(\omega) \ge b_0/2\quad\text{for }k \le N,
\]
for all $h \le h_0$, and $\omega \in [0,1]$.
We can then deduce from \eqref{nuk.lip} that for $\omega<\omega'$,
\[
\begin{split}
\mu_k(\omega') - \mu_k(\omega) &= (\omega'-\omega) \nu_k(\omega) + w'(\nu_k(\omega') - \nu_k(\omega)) \\
&\ge \frac{b_0}4 (\omega'-\omega),
\end{split}
\]
for all $k \le N$, $h \le h_0$ and $\omega \in [0,1]$.
In particular, the functions $\mu_k$ are strictly increasing under these conditions.
Hence the set $\set{ \omega:\>\mu_k(\omega) \in [\mu_0\pm \vep]}$ is an interval for $k \le N$ 
and we can estimate
\[
\abs[\Big]{\set{ \omega : \>\mu_k(\omega) \in [\mu_0\pm \vep]}} \le \frac{8}{b_0} \vep.
\]
For a single $k \le N$ we can thus estimate
\[
\bbP\paren[\big]{\mu_k(\omega_j) \in [\mu_0\pm \vep]\text{ with }\omega_j\ge 0} \le C\vep
\]
where $C$ depends only on $b_0$ and the probability distribution function $g$.
This gives the estimate
\[
\bbP\paren[\big]{\mu_k(\omega_j) \in [\mu_0\pm \vep]\text{ for some }k \le N, \text{ with }\omega_j\ge 0} 
\le CN\vep.
\]
To complete the estimate, we use the Hilbert-Schmidt norm as in the proof of Lemma~\ref{Vhat.lemma},
\[
\begin{split}
N &\le \frac{4}{b_0^2} \sum_{k=1}^N \frac{1}{\nu_k(0)^2} \\
&\le \frac{4}{b_0^2} \norm{\widehat{p}}^2_{\rm HS} \\
&= \frac{8\pi}{hb_0^2} \norm{p}^2_{L^2}.
\end{split}
\]

This yields the probability bound under the restriction $\omega_j \ge 0$. 
As noted at the start, the same argument can be used for $\omega_j \le 0$ by redefining
$\mu_k(\omega)$ as the sequence of negative eigenvalues of $\widehat{a}_0(\omega)$, written in increasing order.
\end{proof}

\begin{lemma}\label{Aw.lemma}
For $\mu_0 \ge b_0 >0$ and $0<\delta<b_0/2$, there exists a constant $C>0$ and a function $r(h) = O(h^\infty)$, each 
depending only on $b_0$ and $v_0$, such that
\[
\bbP\paren[\big]{ \sigma(A_\omega) \cap [\mu_0\pm\delta] \ne \emptyset} \le Ch^{-1}\abs{\Lambda}(\delta+r(h))
\]
\end{lemma}

\begin{proof}
Choose $r(h)$ as in Proposition~\ref{pre.wegner.prop}, so that 
\beq\label{eq:ub1}
\bbP\paren[\Big]{ \sigma(A_\omega) \cap [\mu_0\pm\delta] = \emptyset} 
\ge \bbP\paren[\Big]{ \sigma(\widehat{a}_j(\omega_j)) \cap [\mu_0\pm(\delta + r(h))] = \emptyset\text{ for all }j \in \widehat{\Lambda}}
\eeq
By Lemma~\ref{site.wegner.lemma}, for each $j \in \widehat{\Lambda}$,
\[
\bbP\paren[\Big]{ \sigma(\widehat{a}_j(\omega_j)) \cap [\mu_0\pm(\delta + r(h))] \ne \emptyset} \le Ch^{-1}(\delta+r(h)).
\]
Therefore, it follows from this and the inequality \eqref{eq:ub1} that 
\[
\begin{split}
\bbP\paren[\Big]{ \sigma(A_\omega) \cap [\mu_0\pm\delta] = \emptyset} 
&\ge {\paren[\big]{1-Ch^{-1}(\delta+r(h))}^{\abs{\Lambda}}} \\
&\ge 1 - Ch^{-1}\abs{\Lambda}(\delta+r(h)).
\end{split}
\] 
\end{proof}

\begin{proof}[Proof of Theorem~\ref{thm:wegner_h}]
The existence of an eigenvalue of $H_\omega$
in the range $B_n+I$ corresponds, by Theorem~\ref{bellissard.thm}, to the fact that $Q_{V_\omega}(\mu)$ has an eigenvalue 
$0$ for some $\mu \in I$. 
By Proposition~\ref{Qexpand.prop} there exists a function $c(h) = O(h^3)$ such that
if $0 \in \sigma(Q_{V_\omega}(\mu))$ then $A_\omega$ has an eigenvalue in the range $[\mu \pm c(h)]$. This gives an estimate
\[
\bbP\paren[\big]{0 \in \sigma(Q_{V_\omega}(\mu))\text{ for some }\mu \in I} \le \bbP\paren[\big]{\sigma(A_\omega) \cap [I\pm c(h)] \ne \emptyset}.
\]
For $h \le h_0$, Lemma~\ref{Aw.lemma} then gives
\[
\bbP\paren[\big]{0 \in \sigma(Q_{V_\omega}(\mu))\text{ for some }\mu \in I} \le Ch^{-1} 
\abs{\Lambda}(\abs{I} + O(h^3)).
\]
\end{proof}

\subsection{Semiclassical Minami Estimate}\label{minami.sec}

Given a potential $v_0$, we define $a_j(\omega)$ as in \eqref{aw.def}. As noted above, Lemma~\ref{aj.spec.lemma} implies that the 
operators $\widehat{a}_j(\omega)$ have spectrum independent of $j$. Define the symbol $p_0$ as in \eqref{p0.def}, so that
$\widehat{a}_0(\omega) =\omega \widehat{p}_0 + O(h^2)$.
For the Minami estimate, we require the following:
\begin{itemize}
\item\emph{spectral gap assumption}: Assume that for $b_0>0$, there exists $\kappa>0$ such that
the eigenvalues of $\widehat{p}_0$ in $\set{\abs{\mu} \ge b_0}$ are simple and separated by gaps of at least $\kappa h$ for all $h \le h_0$.
\end{itemize}
In Appendix~\S\ref{radial.sym.sec} we will discuss a family of examples that satisfy this hypothesis. In particular, one can define $v_0$ by taking a compactly supported cutoff of a radial Gaussian potential.

\begin{proof}[Proof of Theorem \ref{thm:minami_h}]
Suppose that $H_\omega$ has at least two eigenvalues in the range
$B_n+I$, counted with multiplicity. By Theorem~\ref{bellissard.thm} this corresponds to the number of $0$ eigenvalues of 
$Q_{V_\omega}(\mu)$, for $\mu \in I$. 
By Proposition~\ref{Qexpand.prop} there exists a function $c(h) = O(h^3)$ such that
\beq\label{Aw.est}
\bbP \paren[\Big]{\#\paren[\big]{\sigma(H_{V_\omega}) \cap (B_n + I)} \ge 2} \le 
\bbP\paren[\Big]{\#\paren[\big]{\sigma(A_\omega) \cap [I\pm c(h)]} \ge 2}.
\eeq

Under the spectral gap assumption, for $I \subset [b_0,1]$ and $\kappa h \le \abs{I}$, 
the site operators $\widehat{a}_j(\omega_j)$ have at most a single simple eigenvalue in $I$. 
Under this assumption, the contrapositive of the statement in
Proposition~\ref{pre.minami.prop} then implies that, by adjusting $c(h)$ by an $O(h^\infty)$ correction if needed,
\begin{equation}\label{pre.minani.eq}
\bbP \paren[\Big]{\#\paren[\big]{\sigma(H_{V_\omega}) \cap (B_n + I)} \ge 2} \le 
\bbP\paren[\bigg]{\sum_{j\in \widehat{\Lambda}} \# \paren[\big]{\sigma(\widehat{a}_j(\omega_j)) \cap [I \pm c(h)]} \ge 2}.
\end{equation}

Define the probability $\rho$ so that
\[
\bbP\paren[\big]{\sigma(\widehat{a}_j)(\omega_j) \cap [I \pm c(h)] = k} = \begin{cases}\rho, &k=1,\\ 1-\rho, &k=0, \end{cases}
\]
which is independent of $j$.
Estimating as in Lemma~\ref{site.wegner.lemma} gives
\begin{equation}\label{q.bound}
\rho \le Ch^{-1}(\abs{I} + 2c(h)).
\end{equation}
Then, because the $\omega_j$ are independent,
\[
\begin{split}
\bbP\paren[\bigg]{\sum_{j\in \widehat{\Lambda}} \# \paren[\big]{\sigma(\widehat{a}_j(\omega_j)) \cap [I \pm c(h)]} \le 1}
&= (1-\rho)^m + m\rho(1-\rho)^{m-1} \\
&\ge 1 - {m\choose 2}\rho^2,
\end{split}
\]
where $m = \abs{\Lambda}$. Combining this estimate with \eqref{Aw.est} and \eqref{pre.minani.eq} gives
\[
\bbP \paren[\Big]{\#\paren[\big]{\sigma(H_{V_\omega}) \cap (B_n + I)} \ge 2} \le {m\choose 2}\rho^2,
\]
and the result follows from \eqref{q.bound}.
\end{proof}

\section{Additional Wegner estimates}\label{scaling.sec}

In this section, we first present a simple proof of the Wegner estimate in small energy intervals at the band edges which yields the correct scaling with respect to the volume (surface area) and energy interval size. This result implies the Lipschitz continuity of the density of states measure.  We then discuss the technique for proving the Wegner estimate introduced by Wang in \cite{wm_wang97}.
The idea is essentially to convert a shift in the spectral parameter into a rescaling of the random variables. The 
principal benefit of this approach is that it does not require positivity or support restrictions on the single-site potentials. The drawback is that the volume dependence of the Wegner estimate is not optimal (see, however, Proposition \ref{prop:wegner_edge}).

The form of the random potential $V_\omega$ is as described in \S\ref{intro.sec}. However, Wang's argument
does not require disjoint supports of the single-site potentials. In this section, we loosen the assumption to 
$v_0 \in \calS(\bbR^2)$, with
\begin{equation}\label{vj.overlap}
\sum_{j \in \bbZ^2} \abs{v_j} \le 1.
\end{equation}
On the other hand, this approach does require smoothness of the single-site probability density. That is, in this section, we take $d\bbP(\omega_j) = g(\omega_j) d\omega_j$ with $g \in \cinf[-1,1]$.

In \S\ref{subsec:wegner_edge1}, we present a simple proof of the Wegner estimate for energies within a small distance of the band edge. This estimate has the optimal dependence of $|\Lambda|$. Thus, it implies the local. Lipschitz continuity of the density of states.

In \S\ref{subsec:wang_wegner1}, we revisit Wang's proof on the Wegner estimate. 
We include an independent proof of Wang's version of the Wegner estimate \cite[Prop.~3.1]{wm_wang97}
for several reasons. First of all, the uniformity of certain estimates
with respect to the volume $\abs{\Lambda}$ is a crucial point that was not made transparent in the original proof. Secondly, we 
observe that the universal upper bound on the counting function, \cite[Lemma~3.1]{wm_wang97},
admits a much shorter and more elementary proof. Finally, there is a slight mistake in the proof of 
\cite[Prop.~3.1]{wm_wang97}. This issue, translated to our notation, is that the number of eigenvalues of $H_{V_\omega}$ greater than
$(2n+1) + \mu_0$ was assumed to be equal to the number of eigenvalues of $Q_{V_\omega}(\mu_0)$ greater than $0$.
By Theorem~\ref{bellissard.thm}, we instead need to count the number of $\mu \ge \mu_0$ for which 
$0 \in \sigma(Q_{V_\omega}(\mu))$. The correction requires an additional error estimate, but does not alter the
strategy of the proof.

\subsection{Wegner estimate at the band edge}\label{subsec:wegner_edge1}

To illustrate the scaling philosophy, we start with a very simple estimate near the edge of the band.
This argument uses the original Landau Hamiltonian and does not require the Grushin method.

\begin{proposition}\label{prop:wegner_edge}
For $\epsilon > 0$,
\[
\bbP\paren[\Big]{B_n+ \mu \in \sigma(H_{V_\omega})\text{ for some }\abs{\mu} \ge 1-\epsilon}
\le C\abs{\Lambda}\epsilon,
\]
where $C$ depends only on $g$.
\end{proposition}

\begin{proof}
By the assumption that $\abs{v_0}\le 1$, 
\[
d\paren[\big]{\sigma(H_{V_\omega}), \sigma(H_{0})} \le \max_{j \in \widehat{\Lambda}}\, \abs{\omega_j}
\]
Therefore,
\[
\bbP\paren[\Big]{B_n + \mu \in \sigma(H_{V_\omega})\text{ for some }\abs{\mu} \ge 1-\epsilon}
\le \bbP\paren*{\max_{j \in \widehat{\Lambda}}\, \abs{\omega_j} \ge 1-\epsilon}.
\]
Since the $\omega_j$ are independent variables,
\[
\bbP\paren*{\max_{j \in \widehat{\Lambda}}\, \abs{\omega_j} \ge 1-\epsilon}
= 1 - \paren*{\int_{-1+\epsilon}^{1-\epsilon} g(\omega) d\omega}^{\abs{\Lambda}}.
\]
Since $g$ is a continuous probability density,
\[
\int_{-1+\epsilon}^{1-\epsilon} g(w) \>dw \ge 1 - c\epsilon,
\]
where $c := 2\sup \abs{g}$. This gives
\[
\begin{split}
\bbP\paren*{\max_{j \in \widehat{\Lambda}}\, \abs{\omega_j} \ge 1-\epsilon}
&\le 1- (1-c\epsilon)^{\abs{\Lambda}} \\
&\le c\abs{\Lambda} \epsilon
\end{split}
\]
provided $c\epsilon < 1$. (The result is trivial for $c\epsilon \ge 1$.)
\end{proof}

\begin{remark}
The Wegner estimate at the band edge in Proposition \ref{prop:wegner_edge} has the correct dependence on the area $|\Lambda|$. As in \cite[Theorem 2.2]{ch96}, this implies the integrated density of states is Lipschitz continuous for $E \in (B_n - \epsilon, B_n)$. This is the first result on the density of states for the random Landau Hamiltonian with nonsign-definite, single-site potentials.

\end{remark}

\subsection{Wang's Wegner estimate}\label{subsec:wang_wegner1}

The main result of this section is the following Wegner estimate adapted from \cite[Prop.~3.1]{wm_wang97},
which was cited as \eqref{eq:wegner_wmw} in the introduction. 
As noted above, we assume here that $g$ is smooth but do not require disjoint supports of the $v_j$. 

\begin{theorem}[Wegner estimate]\label{thm:wegner}
There exists $h_0$ depending on $v_0$ such that for $I \subset [b_0,1]$ with $b_0>0$ and $h \le h_0$,  
\[
\bbP \paren[\Big]{\#\paren[\big]{\sigma(H_{V_\omega}) \cap (B_n + I)} \ge 1}
\le Ch^{-1} \abs{\Lambda}^2 \abs{I},
\]
where the constant $C$ and the $O(h^3)$ error estimate depend on $n$, $v_0$, $b_0$, $h_0$, and the probability density $g$.
\end{theorem}

Wang's approach to Theorem~\ref{thm:wegner} is based on the approximation of the effective Hamiltonian 
$Q_V(\mu)$ by its leading term $\widehat{V} - \mu$. 
By Proposition~\ref{Qexpand.prop} we can write
\begin{equation}\label{QV.RV}
Q_V(\mu) = \widehat{V} - \mu +  F_V(\mu),
\end{equation}
with the remainder estimate
\begin{equation}\label{RV.est1}
\norm{F_V(\mu)} \le Ch,
\end{equation}
where $C$ and $h_0$ depend only $\norm{V}_{C^r_{\rm b}}$ for some $r$. 
For the proof of Theorem~\ref{thm:wegner}, we need to consider a more explicit dependence that $F$ has on
$\mu$ and $V$.

\begin{lemma}\label{RV.est.lemma}
Assume $V \in S(\bbR^2)$ as defined in \eqref{symS.def}, with $\abs{V} \le 1$.
For $h \le h_0$ the remainder term in \eqref{QV.RV} satisfies 
\begin{equation}\label{RV.est2}
\norm{F_V(\mu) - F_V(\mu_0)} \le Ch^3\abs{\mu-\mu_0}.
\end{equation}
and
\begin{equation}\label{RV.est3}
\norm{F_{e^tV}(\mu) - F_V(\mu)} \le Cht
\end{equation}
for $\abs{t}\le 1$, where $C$ and $h_0$ depend only $n$ and $\norm{V}_{C^r_{\rm b}}$ for some $r$.
\end{lemma}
\begin{proof}
The estimate \eqref{RV.est2} follows from the analysis in  the proof of 
Proposition~\ref{Qexpand.prop}, together with the observation that
\[
\norm{E_0(\mu) - E_0(\mu_0)} \le 2h\abs{\mu-\mu_0}
\]
for $h \le 1$ and $\mu,\mu_0 \in [-1,1]$. The estimate \eqref{RV.est3} similarly follows from an inspection of 
the $V$ dependence of the error terms in Proposition \ref{Qexpand.prop}.
\end{proof}  

Our replacement for \cite[Lemma~3.1]{wm_wang97} is the following simple spectral estimate.
\begin{lemma}\label{Vhat.lemma}
Suppose that $V \in L^2(\bbR^2) \cap S(\bbR^2)$. Then,
for any $b_0 > 0$,
\[
\#\set*{\mu \in \sigma(\widehat{V}) :\> \mu \ge b_0} \le \frac{2\pi}{hb_0^2} \norm{V}^2_{L^2(\bbR^2)}.
\]
\end{lemma}

\begin{proof}
Since $\widehat{V}$ is self-adjoint, its Hilbert-Schmidt norm is given by
\[
\norm{\widehat{V}}^2_{\rm HS} = \sum_{\mu \in \sigma(\hat{V})} \mu^2.
\]
Restricting the sum to $\mu \ge b_0$ gives the estimate 
\[
\#\set*{\mu \in \sigma(\widehat{V}) :\> \mu \ge b_0} \le \frac{1}{b_0^2} \norm{\widehat{V}}^2_{\rm HS}.
\]
To compute the Hilbert-Schmidt norm we note that $\widehat{V}$ has the integral kernel
\[
K(y,y') = \frac{1}{h} \int_{\bbR} e^{i(y-y')\eta/h} V(\tfrac{y+y'}2,\eta) \>d\eta.
\]
Thus by a simple change of variables and an application of Plancherel's theorem, 
\[
\begin{split}
\norm{\widehat{V}}^2_{\rm HS} &= \norm{K}^2_{L^2(\bbR^2)} \\
&= \frac{2\pi}{h} \norm{V}^2_{L^2(\bbR^2)}.
\end{split}
\]
\end{proof}

\begin{proof}[Proof of Theorem~\ref{thm:wegner}]
The scaling aspect of the proof makes it convenient to write the interval in the form
\[
I = \sqbrak*{\mu_0e^{-\delta},\mu_0e^{\delta}}.
\]
In the notation of \eqref{QV.RV}, let 
\[
T_\omega(\mu) := \widehat{V}_\omega + F_{V_\omega}(\mu),
\]
By Theorem~\ref{bellissard.thm}, $B_n + \mu$ is an eigenvalue of $H_{V_\omega}$ with $\abs{\mu}\le1$ if and only if 
$\mu$ is an eigenvalue of $T_\omega(\mu)$.
Assume that $\mu$ satisfies these conditions with $\mu\in [\mu_0e^{-\delta},\mu_0e^{\delta}]$,
We can then estimate the resolvent of $T_{\omega}(\mu_0)$ by 
\[
\norm*{(T_{\omega}(\mu_0)-\mu)^{-1}} \ge \frac{1}{\norm*{F_{V_\omega}(\mu_0) - F_{V_\omega}(\mu)}}.
\]
By the spectral theorem this implies that
$T_{\omega}(\mu_0)$ has at least one eigenvalue in the range 
\[
\mu \pm \norm*{F_{V_\omega}(\mu_0) - F_{V_\omega}(\mu)}.
\]
By the bound \eqref{RV.est2}, for $h \ge h_0$ we have
\begin{equation}\label{Fmu.est}
\norm*{F_{V_\omega}(\mu_0) - F_{V_\omega}(\mu)} \le Ch^3\abs{\mu_0-\mu},
\end{equation}
where the constant $C$ depends only on the sup norms of $v_0$ and not on $\omega$ or any other 
parameters.

By choosing $h_0$ sufficiently small, depending only on $v_0$, 
we can guarantee that the constant in \eqref{Fmu.est} satisfies $Ch_0^3 \le e^{-2}$.
Thus for $h \le h_0$, $0<\delta<1$, and $\mu$ as above, it follows that 
$T_{\omega}(\mu_0)$ has at least one eigenvalue within the interval 
$[\mu_0e^{-2\delta},\mu_0e^{2\delta}]$.

For $b \in (0,1)$ and $\omega$ as above, set 
\[
N(b, \omega) = \#\set*{\gamma \in \sigma(T_{\omega}(\mu_0)): \gamma \ge b}.
\]
To summarize the arguments made above, 
the existence of $\mu\in (0,1]$ such that $B_n + \mu$ is an eigenvalue of $H_{V_\omega}$ 
is equivalent to $0 \in \sigma(Q_{V_{\omega}}(\mu))$.
We have shown that if this holds for $\mu\in [\mu_0e^{-\delta},\mu_0e^{\delta}]$ then
\begin{equation}\label{N2d}
N(\mu_0e^{-2\delta}, \omega) - N(\mu_0e^{2\delta}, \omega) \ge 1.
\end{equation}
Since the left side of \eqref{N2d} is a positive, integer-valued random variable, we can thus estimate
\[
\bbP\paren[\Big]{\exists \mu \in \sqbrak*{\mu_0e^{-\delta},\mu_0e^{\delta}}:\> \mu \in \sigma(T_\omega(\mu))}
\le \bbE\paren[\Big]{N(\mu_0e^{-2\delta}, \omega) - N(\mu_0e^{2\delta}, \omega)}.
\]

Our goal is to move the scaling factors $e^{\pm 2\delta}$ from $\mu_0$ onto the vector $\omega$. Observe that
\[
\begin{split}
T_{e^t \omega}(\mu_0) &= e^t \widehat{V}_\omega + F_{V_{e^t\omega}}(\mu_0) \\
&= e^tT_{\omega}(\mu_0) - e^tF_{V_\omega}(\mu_0) + F_{e^tV_\omega}(\mu_0).
\end{split}
\]
By the estimates from Lemma~\ref{RV.est.lemma} and equation (\ref{RV.est1}), the remainder terms are bounded by
$Ch \abs{t}$, where $C$ depends only on $v_0$.
For $h\le h_0$, with $h_0$ depending now on both $v_0$ and $b_0$, we can assume that
\[
\norm*{T_{e^t \omega}(\mu_0) - e^t T_{\omega}(\mu_0)} \le \mu_0(1-e^{-\delta})
\]
for $\abs{t} \le 3\delta$.
By standard perturbation theory, this implies that 
\[
\begin{split}
N(\mu_0, e^{-3\delta}\omega)
&\le \#\set*{\gamma \in e^{-3\delta}\sigma(T_{\omega}(\mu_0)): \gamma \ge \mu_0e^{-\delta}} \\
&= N(\mu_0e^{2\delta}, \omega).
\end{split}
\]
Similarly,
\[
\begin{split}
N(\mu_0, e^{3\delta}\omega)
&\ge \#\set*{\gamma \in e^{3\delta}\sigma(T_{\omega}(\mu_0)): \gamma \ge \mu_0e^{\delta}} \\
&= N(\mu_0e^{-2\delta}, \omega).
\end{split}
\]
Hence, for $h \ge h_0$ and $\mu_0 \ge b_0$,
\begin{equation}\label{N.ineq}
N(\mu_0e^{-2\delta}, \omega) - N(\mu_0e^{2\delta}, \omega) \le N(\mu_0, e^{3\delta}\omega) - N(\mu_0, e^{-3\delta}\omega).
\end{equation}

Let $m = \abs{\Lambda}$ and write $\omega$ as a vector $(\omega_1, \dots \omega_m)$. For $t <0$ we can estimate
\[
\bbE(N(\mu_0, e^{t}\omega) - N(\mu_0, e^{-t}\omega)) \\
= \int_{[-1,1]^{m}} \sqbrak*{N(\mu_0, e^{t}\omega) - N(\mu_0, e^{-t}\omega)} G(\omega) d^m\omega,
\]
where
\[
G(\omega) := \prod_{j=1}^m g(\omega_j).
\]
A change of variables for each term gives
\begin{equation}\label{EN.scale}
\begin{split}
&\int_{[-1,1]^{m}} \sqbrak[\Big]{N(\mu_0, e^{t}\omega) - N(\mu_0, e^{-t}\omega)} G(\omega) d\omega \\
&\qquad= \int_{[-e^t,e^t]^{m}} N(\mu_0,\omega) \sqbrak[\Big]{e^{-mt} G(e^{-t}\omega) -  e^{mt} G(e^{t}\omega)}d^m\omega 
\end{split}
\end{equation}
By Lemma~\ref{Vhat.lemma} and the bound \eqref{RV.est1} we can estimate for $h \ge h_0$ and $\mu_0 \ge b_0$,
\[
\begin{split}
\sup_{\omega} N(\mu_0,\omega) &\le \#\set*{\mu \in \sigma(\widehat{V}_\omega) :\> \mu \ge \mu_0e^{-\delta}} \\
&\le \frac{2\pi e^{2\delta}}{hb_0^2} \norm{V_\omega}_{L^2(\bbR^2)}^2.
\end{split}
\]
Since $\abs{\omega} \le 1$ we have a bound
\[
\norm{V_\omega}_{L^2(\bbR^2)}^2 \le m \norm{v_0}_{L^2(\bbR^2)}^2,
\]
so that for $\delta \in (0,1)$,
\begin{equation}\label{supN}
\sup_{\omega} N(\mu_0,\omega) \le C\frac{m}{hb_0^2}
\end{equation}
where $C$ depends only on $v_0$.

To handle the contribution from $G$ in \eqref{EN.scale} we set $g_t(w) = e^tg(e^tw)$ 
and expand the difference of products to obtain
\[
\begin{split}
e^{-mt} G(e^{-t}\omega) -  e^{mt} G(e^{t}\omega) 
&= \prod_{j=1}^m g_{-t}(\omega_j) - \prod_{j=1}^m g_{t}(\omega_j) \\
&= \sum_{j=1}^m g_{-t}(\omega_1) \cdots g_{-t}(\omega_{i-1}) 
\sqbrak[\Big]{g_{-t}(\omega_i) - g_t(\omega_i)} \\
&\hskip.5in\times g_t(\omega_{i+1}) \cdots g_t(\omega_m) 
\end{split}
\]
Since $g$ is a smooth probability density, this gives
\[
\begin{split}
\int_{[-e^t,e^t]^{m}} \abs[\Big]{e^{-mt} G(e^{-t}\omega) -  e^{mt} G(e^{t}\omega)}\>d^m\omega 
&= m\int_{-e^t}^{e^t} \abs[\big]{g_{-t}(w) - g_t(w)} \>dw \\
&\le Cmt,
\end{split}
\]
where $C$ depends only on $g$. Applying this estimate along with \eqref{supN} to \eqref{EN.scale} yields the
estimate 
\[
\bbP \paren[\Big]{N(\mu_0, e^{3\delta}\omega) - N(\mu_0, e^{-3\delta}\omega) \ge 1} \le C \frac{m^2\delta}{hb_0^2}
\]
where $C$ depends on $v_0$ and $g$. In view of the discussion before \eqref{N.ineq}, 
this completes the argument. 
\end{proof}

\appendix


\section{The Grushin method and the effective Hamiltonian}\label{grushin.sec}

In this appendix, we outline the derivation of the effective Hamiltonian \eqref{QV.exp} following \cite{HS89}. As before, we set $h = B^{-1}$.   We start with a key unitary transformation from \cite[Proposition 1.1]{HS89}.
To explain this unitary operator, we recall the harmonic oscillator operator 
$L := -\del_x^2 + x^2$ acting on $L^2(\bbR)$. We use $L\otimes I$ to denote the corresponding
operator acting on $L^2(\bbR^2)$ through the first variable. 
We also recall the quantization $W$ of the potential $V \in C_b^r (\mathbb{R}^2)$, acting on $L^2(\bbR^2)$,  defined in \eqref{W.def}.

\begin{lemma}\label{l:normal_form}
There exists a unitary operator $U$ on $L^2(\mathbb{R}^2)$ such that 
\begin{equation}\label{HV.conj}
U^{-1} H_V U = h^{-1} L\otimes I + W.
\end{equation}
\end{lemma}

The transformation $U$ can be constructed explicitly as the quantization of a sequence of linear symplectomorphisms, but
we will not need those details here. Note that Lemma~\ref{HV.conj} demonstrates the connection between the 
Landau level $n$ and the corresponding eigenstate $\psi_n$ of $L$ defined in \eqref{psil.def}.

The Grushin method is applied to the operator $L\otimes I + hW$ as follows.
Fix a Landau level $n$ and define the family of operators for $\mu \in \bbC$,
\begin{equation}\label{scrP.def}
\mathscr{P}_V(\mu) = \begin{bmatrix} L\otimes I + hW - h(B_n+\mu) & R_n \\ R_n^* & 0 \end{bmatrix},
\end{equation}
acting on $L^2(\bbR^2) \oplus L^2(\bbR)$, where $R_n: L^2 ( \bbR) \to L^2(\bbR^2) \oplus L^2(\bbR)$ is the map associated to $\psi_n$ as in \eqref{Rl.def}.

For $\abs{h\mu} < 2$ the operator $E_0(\mu)$ defined in \eqref{E0.def} commutes with $L\otimes I$ and satisfies
\[
(L\otimes I-hB_n - h\mu)E_0 = 1 - R_nR_n^*.
\]
Since $(L\otimes I-hB_n)R_n = 0$, $E_0(\mu)R_n = 0$, and $R_n^*R_n$ is the identity on $L^2(\bbR)$, it follows that
$\mathscr{P}_0(\mu)$ is invertible with inverse given by
\[
\mathscr{E}_0(\mu) = 
\begin{bmatrix}
E_0(\mu) & R_n \\
R_n^* & h\mu
\end{bmatrix}.
\]
We can write \eqref{scrP.def} as $\mathscr{P}_V = \mathscr{P}_0 + h\mathscr{W}$, where
\[
\mathscr{W} = \begin{bmatrix}
W & 0 \\
0 & 0
\end{bmatrix}.
\]
As noted in \S\ref{eff.H.sec}, we have the bounds $\norm{R_n} = 1$, $\norm{E_0(\mu)} = (2-h\mu)^{-1}$ for 
$\abs{h\mu} < 2$, and $\norm{W} \le C \norm{V}_{C^r_{\rm b}}$.
Therefor there exists $h_0\le 1$ depending only on $\norm{V}_{C^r_{\rm b}}$ 
such that for $h \le h_0$ and we can invert $\mathscr{P}_B$ by taking
\begin{equation}\label{scrE.def}
\mathscr{E}_V(\mu) = \paren[\big]{1+ \mathscr{E}_0(\mu) h \mathscr{W}}^{\!-1} \mathscr{E}_0(\mu),
\end{equation}
for $\abs{\mu} \le 1$.

Using the standard Grushin notation, we write the inverse in block form as
\begin{equation}\label{scrE.block}
\mathscr{E}_V(\mu) = \begin{bmatrix} E(\mu) & E^+(\mu) \\ E^-(\mu) & E^{-+}(\mu) \end{bmatrix}.
\end{equation}
The effective Hamiltonian that we are looking for is then defined by 
\begin{equation}\label{QV.def}
Q_V(\mu) = -h^{-1}  E^{-+}(\mu),
\end{equation}
acting on $L^2(\bbR)$.

\begin{proof}[Proof of Theorem~\ref{bellissard.thm}]
By Lemma~\ref{l:normal_form}, $B_n + \mu \in \sigma(H_V)$ for $\abs{\mu} \le 1$ if and only if
$L\otimes I + hW - h(B_n+\mu)$ fails to be invertible. The Schur complement formula, applied to the block forms  \eqref{scrP.def} and \eqref{scrE.block}, implies that $L\otimes I+ hW - h(B_n+\mu)$ is invertible if and only if 
$Q_V(\mu)$ is invertible. These inverses are given by 
\[
\paren[\big]{L\otimes I + hW - h(B_n+\mu)}^{\!-1} = E(\mu) + h^{-1} E^+(\mu) Q_V(\mu)^{-1} E^-(\mu)
\]
and
\[
Q_V(\mu)^{-1} = - hR_n^* \sqbrak[\Big]{E(z) - \paren[\big]{L\otimes I + hW - h(B_n+\mu)}^{\!-1}} R_n.
\]
It is shown in \cite[Prop.~2.3.1]{HS89} that $Q_V(\mu)$
is a semiclassical pseudodifferential operator with respect to $h$.
\end{proof}

To conclude this section, let us explain how \eqref{scrE.def} leads to the formula \eqref{QV.exp}
for $Q_V(\mu)$.
For $h \le h_0$ the right side of \eqref{scrE.def} can be written as a geometric series
\[
\mathscr{E}_V(\mu) = \sum_{k=0}^\infty (-h)^k (\mathscr{E}_0(\mu) \mathscr{W})^k \mathscr{E}_0(\mu),
\]
which converges in operator norm for $\abs{\mu} \le 1$. A simple computation gives
\[
(\mathscr{E}_0(\mu) \mathscr{W})^k  = \begin{bmatrix} (E_0(\mu)W)^k & 0 \\ R_n^*W(E_0(\mu)W)^{k-1} & 0 
\end{bmatrix}
\]
for $k \ge 1$. It follows that
\[
E^{-+}(\mu) = h\mu + \sum_{k=1}^\infty (-h)^k R_n^*W (E_0(\mu)W)^{k-1} R_n.
\]
Dividing by $-h$ yields \eqref{QV.exp}.

\section{Radial symbol quantizations}\label{radial.sym.sec}

In this section we discuss some examples of potential functions $v_0$ which demonstrate that the spectral 
requirements of \S\ref{minami.sec} can be satisfied.
For this discussion, let us recall the semiclassical Weyl quantization of a symbol $a(y,\eta)$ defined in \eqref{widehat.def},
\[
\widehat{a}f(y) := \frac{1}{2\pi h} \int_{\bbR^2} e^{i(y-y')\eta/h} a(\tfrac{y+y'}2,\eta) f(y')\>dy'\>d\eta.
\]
The quantization of the quadratic polynomial
\[
q(y,\eta) := y^2 + \eta^2,
\]
is the quantum harmonic oscillator 
\[
\widehat{q} = - h^2 \del_y^2 + y^2.
\] 
with eigenvalues $(2k+1)h$ for $k \in \bbN_0$. The corresponding eigenfunctions 
are rescalings of the oscillator states $\psi_k$ introduced in \eqref{psil.def},
\begin{equation}\label{phik.basis}
\phi_k(y) := h^{-\frac14} \psi_k\paren*{h^{-\frac12}y}
\end{equation}
for $k \in \bbN_0$.

\subsection{Mehler's formula}
The quantization of the function $e^{-zq}$ for $z \in \bbC$ yields an operator with a Gaussian kernel function. 
A classical result of Mehler \cite{mehler_1866} gives the expansion of a Gaussian kernel in terms of Hermite polynomials. 
This implies in particular that  the quantization of $e^{-zq}$ is a diagonal operator 
with respect to the basis \eqref{phik.basis}, an observation which has appeared many times in the
literature. Lemmas~\ref{mehler.lemma} and \ref{eig.Psi.lemma} below are adapted from \cite[\S4]{Ho95} and \cite[Lemma~1.1]{Ler19}.

\begin{lemma}\label{mehler.lemma}
With $q(y,\eta) := y^2 + \eta^2$ and  $z \in \bbC$ with for $\re z \ge 0$ and  $hz \ne -1$, we have
\[
\widehat{e^{-zq}} = \sum_{k=0} \frac{(1-hz)^k}{(1+hz)^{k+1}} \Pi_k
\] 
 where $\Pi_k$ denotes the orthogonal projection onto $\phi_k$ defined in \eqref{phik.basis}.
\end{lemma}
\begin{proof}
To keep the notation readable, let us denote the quantization \eqref{widehat.def} by
\[
\Op(a) := \widehat{a}.
\]
For $t \ge 0$ consider the family of operators 
\[
A(t) := \frac{1}{\cosh(ht)} \Op\paren*{e^{-\frac{1}{h} \tanh (ht) q}}.
\]
Differentiation gives
\begin{equation}\label{dotA}
\del_t A(t) = \frac{h \sinh (ht)}{\cosh^2(ht)} A(t) - \frac{1}{\cosh^3(ht)} \Op\paren*{qe^{-\frac{1}{h} \tanh (ht) q}}.
\end{equation}
We can simplify this expression by computing the Moyal product
\[
q \# e^{-cq} =  (1-h^2c^2) qe^{-cq} + h^2c e^{-cq},
\]
which implies that
\[
\Op(q) A(t) = \frac{1}{\cosh^3 (ht)} \Op\paren*{qe^{-\frac{1}{h} \tanh (ht) q}}
- h \frac{\tanh(ht)}{\cosh(ht)} A(t).
\]
In combination with \eqref{dotA} this shows that
\[
\del_t A(t) + \Op(q)A(t) = 0.
\]
Since $A(0) = I$, by integrating we obtain
\[
A(t) = \exp(-t\Op(q)).
\]
Applying $A(t)$ to an element $\phi_k$ of the basis \eqref{phik.basis} thus gives
\begin{equation}\label{Mt.phik}
A(t) \phi_k = e^{-(2k+1)ht}\phi_k.
\end{equation}

Now let us set $z = h^{-1} \tanh(ht)$ and solve for 
\[
ht = \frac{1}{2} \log \paren*{\frac{1+hz}{1-hz}}.
\]
This implies
\[
e^{-(2k+1)ht} = \paren*{\frac{1-hz}{1+hz}}^{\!k+\frac12}
\]
and
\[
\cosh(ht) = (1 - h^2z^2)^{-\frac12}.
\]
Substituting these calculations into \eqref{Mt.phik} gives
\[
\widehat{e^{-zq}} \phi_k = \frac{(1+hz)^k}{(1-hz)^{k+1}}\phi_k,
\]
which proves the result for $0\le hz < 1$. The general result then follows by analytic continuation for $\re z > 0$
and then by continuity to $\re z = 0$.
\end{proof}


\subsection{Radially symmetric symbols}
Let us consider a radial symbol of the form
\begin{equation}\label{psi.radial}
\psi(y,\eta) = \Psi(y^2+\eta^2),
\end{equation}
where $\Psi \in \calS(\bbR)$ is even. The assumption of evenness is convenient because it allows us to use the 
cosine form of the Fourier transform,
\[
\tilde{\Psi}(\tau) := \int_0^\infty \cos (\tau u) \Psi(u)\>du.
\]

\begin{lemma}\label{eig.Psi.lemma}
Suppose $\psi$ is given by \eqref{psi.radial} with $\Psi \in \calS(\bbR)$ even. 
Then $\widehat{\psi}$ is a diagonal operator with respect to the basis \eqref{phik.basis}, 
\[
\widehat{\psi} \phi_k = \lambda_k(h) \phi_k,
\]
where
\begin{equation}\label{eig.Psi}
\lambda_k(h) = \frac{1}{2\pi} \int_{-\infty}^\infty \frac{(1+ih\tau)^k}{(1-ih\tau)^{k+1}} \tilde{\Psi}(\tau)\>d\tau.
\end{equation}
\end{lemma}
\begin{proof}
The inverse Fourier transform formula gives
\[
\widehat{\psi} = \frac{1}{2\pi} \int_{-\infty}^\infty \tilde{\Psi}(\tau) \widehat{e^{i\tau q}} \>d\tau,
\]
with $q(y,\eta) := y^2 + \eta^2$ as above.
Since $\tilde{\Psi} \in \calS(\bbR)$, we can substitute the Mehler formula from 
Lemma~\ref{mehler.lemma}, with $z = - i\tau$, and interchange the integral with the sum over $k$. This yields
\[
\widehat{\psi} = \sum_{k=0}^\infty \lambda_k(h) \Pi_k,
\]
with $\lambda_k$ as given in \eqref{eig.Psi}.
\end{proof}

It seems worth pointing out that applying Plancherel's formula to \eqref{eig.Psi} yields a direct integral 
formula for the eigenvalues,
\[
\lambda_k(h) = \int_0^\infty \Psi(hu) Q_k(u) e^{-u}\>du,
\]
where $Q_k$ is the polynomial
\[
Q_k(u) := \sum_{n=0}^k {k \choose n} \frac{1}{n!} (-1)^{k-n} (2u)^n.
\]
This expression for $\lambda_k(h)$ might be useful for producing examples, but \eqref{eig.Psi} seems better suited for asymptotic analysis.


\subsection{Examples of single-site potentials: Gaussian functions}

Using Lemma~\ref{eig.Psi.lemma} we can construct specific examples of single-site potentials $v_0$
that satisfy the spectral gap assumption of \S\ref{minami.sec} with constant $b_0 \in (0,1)$. 
Let us start by considering a simple Gaussian 
\[
v_0 = e^{-q}, ~~~~q(y, \eta) = y^2 + \eta^2.
\]
Lemma~\ref{mehler.lemma} gives the eigenvalues of $\widehat{v_0}$ as 
\[
\lambda_k(h) = \frac{(1-h)^k}{(1+h)^{k+1}}
\]
for $k \in \bbN_0$. These eigenvalues satisfy $\lambda_k \le e^{-2kh}$. Therefore, a
cutoff of the form $\lambda_k \ge b_0>0$ implies that $kh \le \gamma_0$ where
$\gamma_0 = -\tfrac12\log b_0$. In this range we can estimate
\[
\begin{split}
\log \lambda_k(h) &= k \log(1-h) - (k+1)\log(1+h) \\
&= -(2k+1)h + O(h^2),
\end{split}
\]
with a constant that depends only on $b_0$. Thus, 
\[
\lambda_k(h) = e^{-(2k+1)h} \paren*{1 + O(h^2)},\quad\text{for }kh \le \gamma_0,
\]
which shows that the spectral gap assumption is satisfied by the eigenvalues of $\widehat{v}_0$ itself.

Recall, however, that the spectral gap assumption refers not to $v_0$ itself, but rather to the symbol correction 
defined by \eqref{p0.def}. For $v_0 = e^{-q}$ this corrected symbol satisfies
\[
p_0 = \sqbrak[\Big]{1 + (2n+1)(q-1)h} e^{-q} + O(h^2).
\]
The spectrum is easily worked explicitly out from Lemma~\ref{eig.Psi}, with the eigenvalues of $\widehat{p}_0$ 
given by
\[
\mu_k(h) = \sqbrak[\Big]{1 - (2n+1)h}  \frac{(1-h)^k}{(1+h)^{k+1}} + (2n+1)  (2k+1-h)h^2 \frac{(1-h)^{k-1}}{(1+h)^{k+2}}
+ O(h^2).
\]
Under the assumption that $kh \le \gamma_0$ we obtain the estimate
\begin{equation}\label{mukh}
\mu_k(h) = e^{-(2k+1)h} \paren[\Big]{1 - (2n+1)h + O(h^2)}.
\end{equation}
implying the following:
\begin{lemma}\label{v0.gap.lemma}
The spectral gap assumption of \S\ref{minami.sec} is satisfied for the Gaussian potential $v_0 = e^{-q}$.
\end{lemma}

\subsection{Examples of single-site potentials: Compactly supported and nonsign-definite functions}

To satisfy the requirements for Theorem~\ref{thm:minami_h} we need also for $v_0$ to be compactly
supported. The following result allows us to control the error when a compactly support cutoff function is
included in the potential. 
In general, the eigenvalues given by \eqref{eig.Psi} are not ordered in $k$. However, we can prove in general
that the lower bound $\lambda_k \ge b_0$ implies that $kh \le \gamma_0$.
\begin{lemma}\label{kh.bound.lemma}
For an even function $\Psi \in \calS(\bbR)$ there exists $C>0$ depending only on $\Psi$ such that the eigenvalues given in 
Lemma~\ref{eig.Psi.lemma} satisfy
\[
\lambda_k(h) \le \frac{C}{kh}.
\]
\end{lemma}
\begin{proof}
Define the real phase function 
\[
\phi(\tau) := -i \log \paren*{\frac{1+ih\tau}{1-ih\tau}}, 
\]
so that \eqref{eig.Psi} becomes an integral of Laplace type,
\[
\lambda_k = \frac{1}{2\pi} \int_{-\infty}^\infty e^{ik\phi(\tau)} \frac{\tilde{\Psi}(\tau)}{1-ih\tau}\>d\tau.
\]
The assumptions on $\tilde{\Psi}$ justify the following integration by parts:
\[
\begin{split}
\lambda_k &= \frac{1}{2\pi} \int_{-\infty}^\infty \frac{1}{ik\phi'(\tau)} 
\sqbrak*{\frac{\partial}{\partial\tau} e^{ik\phi(\tau)} }
\frac{\tilde{\Psi}(\tau)}{1-ih\tau}\>d\tau \\
&= -\frac{1}{2\pi} \int_{-\infty}^\infty e^{ik\phi(\tau)} \frac{\partial}{\partial\tau} 
\sqbrak*{\frac{1}{ik\phi'(\tau)}\frac{\tilde{\Psi}(\tau)}{1-ih\tau}}\>d\tau \\
&= \frac{i}{4\pi kh} \int_{-\infty}^\infty e^{ik\phi(\tau)} \frac{\partial}{\partial\tau} 
\sqbrak*{(1+ih\tau) \tilde{\Psi}(\tau)}\>d\tau.
\end{split}
\]
The estimate follows directly.
\end{proof}

\begin{lemma}\label{Psi.cut.lemma}
Suppose that $\Psi$ is an even function in $\calS(\bbR)$ with $\Psi(u) = 0$ for $\abs{u} \le b_0$. 
Then for $(2k+1)h \le b_0$ the eigenvalues of $\widehat{\psi}$ satisfy
\[
\lambda_k(h) = O(h^2),
\]
with a constant that depends only on $\Psi$.
\end{lemma}
\begin{proof}
We start by expanding the log of the integration factor in \eqref{eig.Psi},
\begin{equation}\label{log.iht}
\begin{split}
\log\paren*{\frac{(1+ih\tau)^k}{(1-ih\tau)^{k+1}}} &= k\log(1+ih\tau) - (k+1)\log(1-ih\tau) \\
&= (2k+1)ih\tau + O(k\abs{h\tau}^3) + O(\abs{h\tau}^2),
\end{split}
\end{equation}
For $kh \le \gamma_0$ and $\abs{h\tau} \le 1$ this can be exponentiated to give an estimate 
\[
\frac{(1+ih\tau)^k}{(1-ih\tau)^{k+1}} = e^{i(2k+1)h\tau} + O\paren*{h^2\brak{\tau}^3}.
\]
Applying this to \eqref{eig.Psi} gives
\[
\lambda_k(h) = \frac{1}{2\pi} \int_{\abs{\tau} \le 1/h} \sqbrak*{e^{i(2k+1)h\tau} + O\paren*{h^2\brak{\tau}^3}} \tilde{\Psi}(\tau)\>d\tau
+ \int_{\abs{\tau} > 1/h} O(\brak{\tau}^{-1}) \tilde{\Psi}(\tau)\>d\tau.
\]
Because $\tilde{\Psi}$ is rapidly decreasing, the final integral is $O(h^\infty)$ and 
the second term in the first integral is $O(h^2)$.
The first integral can be expanded into an inverse Fourier transform
\[
\begin{split}
 \frac{1}{2\pi} \int_{\abs{\tau} \le 1/h} e^{i(2k+1)h\tau} \tilde{\Psi}(\tau)\>d\tau
&=  \frac{1}{2\pi} \int_{-\infty}^\infty e^{i(2k+1)h\tau} \tilde{\Psi}(\tau)\>d\tau + O(h^\infty) \\
&= F((2k+1)h) + O(h^\infty).
\end{split}
\]
Putting all of the error estimates together gives
\[
\lambda_k(h) = F((2k+1)h) + O(h^2),
\]
and the result follows from the support assumption on $F$.
\end{proof}

Since the spectral gap assumption requires spacing of order $h$, Lemma~\ref{Psi.cut.lemma} gives us
the following:
\begin{corollary}\label{psi.cut.lemma}
Let $\chi \in \cinf[0,\infty)$ be a cutoff function with $\chi(u) =1$ for $u \le b_0$ and $\chi(u)$ = 0 for $u \ge 1$.
If the single-site potential $v_0 = \Psi(q)$ satisfies the spectral gap assumption of \S\ref{minami.sec}, then so does
$\chi(q)\Psi(q)$.
\end{corollary}

In particular, Lemma~\ref{v0.gap.lemma} and Corollary~\ref{psi.cut.lemma} together show that the spectral gap assumption 
is satisfied by the cutoff Gaussian potential, 
\[
v_0 = \chi(q) e^{-q}.
\] 
Lemma~\ref{Psi.cut.lemma} also demonstrates the existence of examples of single-site potentials $v_0$ with compact support and that are not sign-definite.
Given a positive radial potential satisfying the spectral gap assumption, modifying the potential to include a sign change outside a neighborhood of the origin will not break the gap assumption.


\end{document}